\documentclass[pdflatex,sn-mathphys,Numbered]{sn-jnl}

\pdfoutput=1

\usepackage{graphicx}%
\usepackage{multirow}%
\usepackage{amsmath,amssymb,amsfonts}%
\usepackage{amsthm}%
\usepackage{mathrsfs}%
\usepackage[title]{appendix}%
\usepackage{xcolor}%
\usepackage{textcomp}%
\usepackage{manyfoot}%
\usepackage{booktabs}%
\usepackage{algorithm}%
\usepackage{algorithmicx}%
\usepackage{algpseudocode}%
\usepackage{listings}%
\usepackage{mathtools}
\usepackage{dsfont}
\usepackage{csquotes, bm}
\usepackage{hyperref}
\usepackage{subcaption}
\usepackage{float}
\usepackage{tikz}
\usepackage{pifont,array}
\usepackage{chngcntr}
\usepackage{geometry}
\usepackage{comment}
\usepackage{longtable, multirow, tabularx}
\usepackage{adjustbox}
\usepackage{tikzpeople}
\usetikzlibrary{shapes.arrows}
\usetikzlibrary{patterns}

\newcommand{\op}[1]{\operatorname{#1}}
\newcommand{\im}{\mathrm{i}}

\newcommand{\bb}[1]{\mathbb{#1}}
\newcommand{\bbC}{\mathbb{C}}

\newcommand{\bbZ}{\mathbb{Z}}
\newcommand{\bbR}{\mathbb{R}}
\renewcommand{\cal}[1]{\mathcal{#1}}
\newcommand{\diff}{\mathrm{d}}
\newcommand{\id}{\mathds{1}}
\newcommand{\eul}{\mathrm{e}}
\newcommand{\ul}[1]{\underline{#1}}

\definecolor{blueBand}{RGB}{153,204,255} 
\definecolor{blueBandEdge}{RGB}{47,129,188} 
\definecolor{EdgeState}{RGB}{255,128,16} 
\definecolor{FermiLine}{RGB}{224,93,93} 
\definecolor{nicerGray}{RGB}{164,164,164} 
\definecolor{niceGreen}{RGB}{64,198,77}
\newcolumntype{Y}{>{\centering\arraybackslash}X}
\tikzset{
	pics/.cd,
	vector in/.style args={#1/#2/#3}{
		code={
			\draw[#1] (0,0)  circle (#3);
			\draw[#2] (45:#3) -- (225:#3) (135:#3) -- (315:#3);
		}
	}
}
\tikzset{
	pics/vector out/.style args={#1/#2/#3}{
		code={
			\draw[#1] (0,0)  circle (#3);
			\fill[#2] (0,0)  circle (#3/4);
		}
	}
}

\theoremstyle{thmstyleone}%
\newtheorem{theorem}{Theorem}
\newtheorem{proposition}[theorem]{Proposition}%
\newtheorem{lemma}[theorem]{Lemma}%

\theoremstyle{thmstyletwo}%
\newtheorem{remark}{Remark}%

\theoremstyle{thmstylethree}%
\newtheorem{definition}{Definition}%

\begin{document}

\title[Article Title]{Topology of 2D Dirac operators with variable mass and an application to shallow-water waves}


\author[1]{\fnm{Sylvain} \sur{Rossi}}\email{syrossi@student.ethz.ch}

\author[1]{\fnm{Alessandro} \sur{Tarantola}}\email{ataranto@phys.ethz.ch}

\affil[1]{\orgdiv{Institute for Theoretical Physics}, \orgname{ETH Z\"urich}, \orgaddress{\street{Wolfgang-Pauli-Str. 27}, \city{Z\"urich}, \postcode{8093}, \state{Z\"urich}, \country{Switzerland}}}




\abstract{A Dirac operator on the plane with constant (positive) mass is a Chern insulator, sitting in class D of the Kitaev table. Despite its simplicity, this system is topologically ill-behaved: the non-compact Brillouin zone prevents definition of a bulk invariant, and naively placing the model on a manifold with boundary results in violations of the bulk-edge correspondence (BEC). We overcome both issues by letting the mass spatially vary in the vertical direction, interpolating between the original model and its negative-mass counterpart. Proper bulk and edge indices can now be defined. They are shown to coincide, thereby embodying BEC.

The shallow-water model exhibits the same illnesses as the 2D massive Dirac. Identical problems suggest identical solutions, and indeed extending the approach above to this setting yields proper indices and another instance of BEC.
}

\maketitle

\section{Introduction}\label{sec:Intro}

Dirac's model is almost ubiquitous in modern physics. Originally introduced to describe relativistic electrons \cite{Dirac1928}, it has more recently found applications in condensed matter \cite{Cayssol13, Wehling14}. Band crossings, effectively described by Dirac-type Hamiltonians, appear in Weyl semimetals \cite{Yan17,Armitage18}, graphene \cite{Geim07} and many other platforms \cite{VV14}. When the host material is topological, such Dirac cones become central in describing the gap closing and quantum phase transition. Indeed, the classification of the Clifford algebras associated with these Dirac Hamiltonians, in presence of the symmetries allowed by the ten-fold way \cite{TenfoldWay}, is one of the routes that lead to the derivation of the periodic table of topological insulators and superconductors \cite{KitaevTable,Schnyder08}.

Yet, surprisingly, natural variants of the relativistic model seem impervious to proper topological treatment. For example, introduce a non-zero mass term and place the model on the 2D plane. From a topology perspective, each of the two modifications is problematic in its own way, as we argue below.
	
In the massless setting, a particle is characterized by the sign of its chirality. The picture extends seamlessly to Weyl semimetals, where one assigns to each band crossing the $ \pm 1 $ chirality of its effective description in terms of massless Dirac fermions \cite{Yan17,Armitage18,Witten16}. Introduction of a (constant) mass term spoils this intuition by mixing the different chirality sectors \cite{PeskinQFT}.

If chirality cannot serve as an index for the massive theory, perhaps its topology is captured by the first Chern number \cite{Hartshorne} of its associated Bloch bundle \cite{TKNN,ProdanSB} (as is the case for other Class D insulators, cf.\,\cite{KitaevTable}). This possibility is hindered by our choice of position space $ \bbR^2 $. By inducing a non-compact Brillouin zone, it prevents the application of the Gauss-Bonnet theorem \cite{GaussBonnet}, destroying the topological significance of the Chern number. 

For the reasons above, and to the best of the authors' knowledge, no good (bulk) invariant can be associated to the massive Dirac Hamiltonian on $ \bbR^2 $. This issue may be circumvented by considering the same model on a manifold with boundary (e.g.\,the upper half-plane $ \bbR \times \bbR_+ $). If an edge invariant emerges (typically, the signed number of gapless modes propagating along the boundary \cite{Angelakis14,Grosse86,Hu22}), its bulk-counterpart might be obtained through \textit{bulk-edge correspondence} (BEC for short, see seminal works \cite{Hatsugai93PRL,Hatsugai93PRB}, selected references \cite{SchulzBaldes99, Elgart05, Elbau02} and modern review \cite{Shapiro20}).

The failure of this approach is revealed by inspecting the edge states. As noticed in  \cite{Gruber06} and argued in App.\,\ref{app:BECFailureDirac}, a single edge mode appears only for half of the allowed boundary conditions. We are confronted with two equally grim alternatives: (i) If BEC is enforced, i.e.\,the edge and bulk indices coincide, the latter attains a (nonsensical) dependence on the boundary condition; (ii) If BEC is not enforced, it is immediately disproven because a good bulk invariant (independent of boundary physics) will not match the edge index in at least half of the cases.

The first part of the manuscript is concerned with the claim that, despite the failures above, a topologically tractable system is obtained by introduction of a \textit{soft} boundary. Rather than cutting the model abruptly, one endows the mass with a position dependence, and lets it change sign. Where the mass is zero, the \enquote{local} Hamiltonian is gapless and formation of bound (interface) states is expected. Indeed, a field-theoretical analysis of this system is found in \cite{JackiwRebbi}. The authors discover edge fields with fermion number $ 1/2 $, a feature reminiscent of the situation of App.\,\ref{app:BECFailureDirac}, where edge states appear half of the time.

Various results in the literature go in the direction of our claim. An informal bulk invariant was proposed in \cite{Witten16}, see below Eq.\,(3.22). It was later formalized by the author of \cite{BalInterface}, who moreover introduced bulk and edge indices for a larger class of \textit{interface Hamiltonians} \cite{JackiwRebbi,HK10}, and further studied such systems in \cite{BHW23,Bal22,QB23}. Our novel contributions consist in an intuitive, yet rigorous geometric derivation of the bulk index of \cite{BalInterface}, along with the introduction of a different edge index.

The findings of the first part of this work can thus be summarized as follows. We start by joining the positive-mass 2D Dirac model with its negative-mass counterpart, inspired by the idea that \enquote{fermions always come in pairs} (cf.\,Nielsen-Ninomiya theorem \cite{NNFirst,NNSecond}). In practice, the pairing is performed by writing the mass term as a function of the second component $x_2$ of the two-dimensional position operator $ \ul{x} = (x_1,x_2) $. The mass profile changes sign at $x_2=0$ and saturates to a constant value $ m_\pm \gtrless 0 $ at $ x_2 \to \pm \infty $. Considering both asymptotic Hamiltonians ($ x_2 \to \pm \infty $) at once allows for a compactification of the Brillouin zone. The Chern number of the Bloch bundle constructed upon it is the new bulk index. It is equal to the signed number of states that localize around the $x_2=0$ interface, and such equality is an instance of bulk-edge correspondence.

In the second part of the manuscript, the technology developed for Dirac Hamiltonians is applied to a related model. Topological indices are defined and another instance of bulk-edge correspondence proven.

Recall indeed that the 2D Dirac Hamiltonian describes spin-$ 1/2 $ particles. Its spin-$1$ counterpart happens to coincide with the Hamiltonian of the rotating shallow-water model \cite{DMV17, Iga95}. The latter is a hydrodynamical model derived from Euler's equations and used to describe the dynamics of thin layers of fluid lying on a rotating bottom. The angular velocity $f$ of this rotation plays the same role as the mass in the Dirac setting. The shallow-water model, upon addition of an odd-viscous term, displays an anomalous bulk-edge correspondence \cite{Tauber19,TDV20} or violates it altogether \cite{GJT21} depending on the boundary condition. After defining the bulk index in complete analogy with the Dirac case, we show that no such violation is present in our setting, at least for a specific profile of the variable angular velocity.

The paper is organized as follows. In Section \ref{sec:Setup}, we introduce the Dirac Hamiltonian with constant mass, naively compute its \enquote{Chern number} and argue why it is ill-defined. The alternative model with variable mass is presented and its essential spectrum found. In Sec.\,\ref{sec:Bulk}, we propose a bulk index for the new model and prove it is topological. In Sec.\,\ref{sec:edge}, we define the edge index, compute it and prove its independence from the choice of mass profile. The shallow-water model is introduced in Sec.\,\ref{sec:SW} as the spin-$1$ counterpart of the previously studied Hamiltonian. A bulk index is defined in complete analogy with Section \ref{sec:Bulk}, and the edge index is computed in a simple case. The two coincide. Sec.\,\ref{sec:Conclusion} is finally devoted to conclusions and future prospects. 

\section{Setup: Two-dimensional Dirac Hamiltonians with (non)-constant mass} \label{sec:Setup}

A Dirac Hamiltonian with constant mass is introduced. Some of its spectral properties are listed, and the Bloch bundle associated with its positive-energy band constructed. Naively computing its \enquote{Chern number} yields non-integer values. A related model with non-constant mass is introduced and its essential spectrum specified. We claim it can be equipped with a well-defined Chern number, which serves as bulk invariant. The proof of such claim is the content of Section \ref{sec:Bulk}.

Consider a spin-$1/2$ particle on $ \bbR^2 $, with Hilbert space $ \cal{H} = L^2 (\bbR^2) \otimes \bbC^2 $. Describe its dynamics by a \textit{two-dimensional Dirac Hamiltonian} $ H_{\pm} $, written in terms of Pauli matrices as
\begin{equation}
	H_\pm \coloneqq (p_1, p_2, m_\pm) \cdot \vec{\sigma} =
	\begin{pmatrix}
		m_\pm & - \im \partial_1 - \partial_2 \\
		- \im \partial_1 + \partial_2 & - m_\pm 
	\end{pmatrix} \,, \label{eq:HPlusMinus}
\end{equation}
where $ \vec{\sigma} = (\sigma_1, \sigma_2, \sigma_3) $, $ p_j \coloneqq - \im \partial_j \,, \ (j = 1,2) $ momentum operator in direction $j$ and $ m_\pm \gtrless 0 $ is a positive (negative) constant.

The operators $ H_\pm $ enjoy a particle-hole symmetry $ C H_\pm C^{-1} = - H_\pm $, where $ C = \sigma_1 K $ and $K$ denotes complex conjugation. Neither time-reversal nor chiral symmetry are present. The model therefore sits in class $D$ w.r.t.\,the Kitaev table \cite{KitaevTable}, and a $ \bbZ $-valued invariant is expected. The Hamiltonians are moreover translation invariant, and their Fourier transform reads
\begin{equation}
	H_\pm \coloneqq (k_1, k_2, m_\pm) \cdot \vec{\sigma} = \vec{d}_\pm (\ul{k}) \cdot \vec{\sigma} \,, \label{eq:HPlusMinusFourier}
\end{equation}
where $ \ul{k}=(k_1,k_2) \ni \bbR^2 $ is a point in the non-compact Brillouin zone $ \bbR^2 $, dual of position space $ \bbR^2 \ni \ul{x} = (x_1,x_2) $. The spectra are purely essential, and consist of two bands 
\begin{equation}
	\pm \omega_+ (\ul{k}) = \pm | \vec{d}_+ (\ul{k}) | = \pm \sqrt{k^2 + m_+^2}
\end{equation}
for $ H_+ $ and 
\begin{equation}
	\pm \omega_- (\ul{k}) = \pm | \vec{d}_- (\ul{k}) | = \pm \sqrt{k^2 + m_-^2}
\end{equation}
for $ H_- $, with $ k \coloneqq \sqrt{k_1^2 + k_2^2} $.

To discuss the topology of $ H_\pm $, one has to associate a Bloch bundle to their bands. We can w.l.o.g.\,restrict our discussion to the top (positive) bands, because the bottom ones are their symmetric counterpart under particle-hole conjugation. Consider the flattened Hamiltonians
\begin{equation}
	H'_\pm (\ul{k}) = \vec{e}_\pm (\ul{k}) \cdot \vec{\sigma} \coloneqq \frac{\vec{d}_\pm (\ul{k})}{| \vec{d}_\pm (\ul{k}) |} \cdot \vec{\sigma} = \frac{1}{\sqrt{k^2 + m_\pm^2}} (k_1, k_2, m_\pm) \cdot \vec{\sigma} \label{eq:HPlusFlat}
\end{equation}
and use $ \vec{e}_\pm $ to construct the projections
\begin{equation}
	P_\pm (\ul{k}) =  \frac{ \id + (\vec{e}_\pm (\ul{k}) \cdot \vec{\sigma}) }{2} \label{eq:SpecProj}
\end{equation}
onto the positive bands. Their associated Bloch bundles then read
\begin{equation}
	\cal{E}_\pm \coloneqq \{ (\ul{k}, \psi_{\ul{k}}) \ | \ \ul{k} \in \bbR^2 \,, \ \psi_{\ul{k}} \in \op{ran} (P_\pm (\ul{k})) \subset \bbC^2 \} \,.
\end{equation}
If $ \cal{E}_\pm $ had a compact base space $ \Omega $ in place of $ \bbR^2 $, its (rightful) Chern number could be computed by (Prop.\,2.1 in \cite{GJT21})
\begin{equation}
	Ch (\cal{E}_\pm) = \frac{1}{2 \pi} \int_{\Omega} \vec{e}_\pm (\ul{k}) \cdot (\partial_{k_1} \vec{e}_\pm (\ul{k}) \wedge \partial_{k_2} \vec{e}_\pm (\ul{k})) \diff k_1 \diff k_2 \,. \label{eq:ChernIntegral}
\end{equation}
In our current setup, the outcome of the integration in Eq.\,\eqref{eq:ChernIntegral} is not a topological invariant, and thus not necessarily an integer. We nonetheless compute it, but denote it by $ C \hbar ( \cdot ) $ to distinguish it from well-defined Chern numbers
\begin{equation}
	C \hbar (\cal{E}_\pm) = \pm \frac{1}{2} = \frac{1}{2} \op{sgn} (m_\pm) \,. \label{eq:FakeChernPlusMinus}
\end{equation}
Evocatively,
\begin{equation}
	C \hbar (\cal{E}_+) - C \hbar (\cal{E}_-) = +1 \label{eq:ChernDifference}
\end{equation}
is a non-zero integer. This is the intuition appearing e.g.\,in \cite{Witten16,HK10}. We make the stronger statement that such a difference is a well-defined topological index, and more specifically the bulk invariant of the following Hamiltonian
\begin{equation}
	H \coloneqq (p_1, p_2, m) \cdot \vec{\sigma} \,, \label{eq:H}
\end{equation}
where $ m $ is a function of the position operator in direction 2 with profile $ m(x_2) $ satisfying:
\begin{enumerate}
	\item $ m(x_2) $ differentiable with continuous derivative;
	
	\item $ m(x_2) $ monotonous for all $x_2$;
	
	\item $ m(x_2) \to m_\pm $ as $ x_2 \to \pm \infty $.
\end{enumerate}
Without loss of generality, we shall also assume $ m(0)=0 $. For definiteness, we pick $ m' (x_2) \geq 0 \ \forall x_2$ (monotonically increasing) here and in Section \ref{sec:Bulk}. By contrast, both $ m' (x_2) \geq 0 $ and $m' (x_2) \leq 0 $ will be considered in Sec.\,\ref{sec:edge}. In any event, flipping the sign of this derivative amounts to flipping the sign of the (to be defined) topological indices.

Notice that translation invariance in direction $ x_1 $ is not lost, and $ H $ can hence be written fiber-wise as
\begin{equation}
	H(k_1) = (k_1, - \im \partial_2, m) \cdot \vec{\sigma} = 
	\begin{pmatrix}
		m & k_1 - \partial_2 \\
		k_1 + \partial_2 & -m
	\end{pmatrix} \,. \label{eq:HFiber}
\end{equation}

The essential spectrum $ \sigma_{\mathrm{e}} (H(k_1)) $ of the fibered operator reads
\begin{equation}
	\sigma_{\mathrm{e}} (H(k_1)) = \left\{ \omega \in \bbR: | \omega | \geq \sqrt{k_1^2 + \min \{ m_-^2, m_+^2 \} } \right\} \,, \label{eq:EssSpecH}
\end{equation}
and this result follows from Thm.\,3.11 in \cite{SchroOps} and
\begin{equation}
	(H - H_\pm) T_a \overset{s}{\to} 0 \,, \qquad (a \to \pm \infty) \,, \label{eq:StrongConv1}
\end{equation}
where $ T_a \coloneqq \eul^{- \im p_2 a} $ and $ \overset{s}{\to} $ denotes strong convergence. Eq.\, \eqref{eq:StrongConv1} says that $H$ \enquote{reduces} to the Hamiltonians $ H_\pm $ above when $ x_2 \to \pm \infty $. Given this fact, \eqref{eq:EssSpecH} embodies the common wisdom that the essential spectrum is determined by \enquote{what happens very far away}.

\begin{remark}
Systems like our $H$ have been extensively studied before \cite{BHW23,QB23,Bal22,BalInterface,HK10,JackiwRebbi}. Their appeal is justified below. Intuitively, $ H_\pm $ are two insulators in the same symmetry class but in different topological phases, as suggested by Eqs.\,\eqref{eq:FakeChernPlusMinus}. The global Hamiltonian $ H $ coincides with $ H_\pm $ for $ x_2 \to \pm \infty $, and spatially interpolates between the two for finite values of $ x_2 $. At $ x_2 = 0 $, $ m(0)=0 $ and the \enquote{local} Hamiltonian is gapless. The gap closing hints at a quantum phase transition, and the system described by $H$ can thus be seen as two different topological insulators smoothly glued together along the $ x_2 = 0 $ line. The invariant associated to a system of this kind (see e.g.\,Ref.\,\cite{HK10}, discussion above Eq.\,(7)) is either the signed number of bound states propagating along the $x_2 = 0$ interface (edge), or the difference between the Chern numbers of the two insulators (bulk), just as proposed in Eq.\,\eqref{eq:ChernDifference}.
\end{remark} 

\section{Bulk index} \label{sec:Bulk}

The notion of \enquote{bulk} is intimately linked to translation invariance: An observer finds itself in the bulk of a material when he cannot perceive any edges or interfaces nearby, i.e.\,the Hamiltonian he is subject to does not change if he moves around slightly. In this sense, the system described by $H$ has two separate bulk regions $ x_2 \to \pm \infty $ where the relevant Hamiltonians are $ H_\pm $, respectively. How to combine this doubled bulk picture into a single, coherent one can be understood with a story.

Imagine position space $ \bbR^2 $ as a translucent sheet of paper. Draw on it the $x_2=0$ \textit{equatorial} line and one reference frame for each half-plane ($x_2 > 0$ and $ x_2 <0 $). Now, fold along the equator, keeping the $ x_2 > 0 $ half-plane facing upwards. Think of an observer (Alice) living on the top layer, near the $ x_2 = 0 $ edge. Have her move away from it in the positive $ x_2 $-direction, until the equator is no longer visible and she experiences constant mass $m_+$. She has reached the \enquote{upper} bulk, where physics is governed by the translation invariant Hamiltonian $H_+$. However, looking down through the translucent paper, Alice will be able to see another plane with \textit{opposite} orientation of the $x_2$-axis (due to the folding), negative mass $m_-$ and Hamiltonian $ H_- $. Her journey is graphically recounted in Fig.\,\ref{fig:Alice}. 
\begin{figure}[h]
	\centering
	\begin{adjustbox}{max size={.95\textwidth}{.8\textheight}}
		\begin{tikzpicture}
			\draw (-16,1.5) -- (-10,1.5);
			\draw (-11.5,-3) -- (-5.5,-3) node[style=midway] (a4) {};
			\draw (-16,1.5) -- (-11.5,-3) node[style=midway] (a1) {};
			\draw (-10,1.5) -- (-5.5,-3) node[style=midway] (a2) {};
			\draw[very thick] (a1.center) -- (a2.center) node[style=midway] (a3) {};
			
			\draw (-4.5,1.5) -- (1.5,1.5);
			\draw (-4.5,0) -- (-3,0);
			\draw[dashed] (-3,0) -- (1.5,0);
			\draw (-4.5,0) -- (-1.5,-3);
			\draw[dashed] (1.5,0) -- (4.5,-3);
			\draw (-4.5,1.5) -- (-1.5,-1.5);
			\draw (1.5,1.5) -- (4.5,-1.5);
			
			\draw (-1.5,-1.5) .. controls (-1,-2) and (-1,-3.5) .. (-1.5,-3) node[style=midway] (c1) {} node[pos=0.85] (t1) {};
			\draw (4.5,-1.5) .. controls (5,-2) and (5,-3.5) .. (4.5,-3) node[style=midway] (c2) {};
			\draw[very thick] (c1.center) -- (c2.center);
			
			\draw (-4.5+11.5,1.5) -- (1.5+11.5,1.5);
			\draw (-4.5+11.5,0) -- (-3+11.5,0);
			\draw[dashed] (-3+11.5,0) -- (1.5+11.5,0);
			\draw (-4.5+11.5,0) -- (-1.5+11.5,-3);
			\draw[dashed] (1.5+11.5,0) -- (4.5+11.5,-3);
			\draw (-4.5+11.5,1.5) -- (-1.5+11.5,-1.5);
			\draw (1.5+11.5,1.5) -- (4.5+11.5,-1.5);
			
			\draw (-1.5+11.5,-1.5) .. controls (-1+11.5,-2) and (-1+11.5,-3.5) .. (-1.5+11.5,-3) node[style=midway] (c3) {};
			\draw (4.5+11.5,-1.5) .. controls (5+11.5,-2) and (5+11.5,-3.5) .. (4.5+11.5,-3) node[style=midway] (c4) {};
			\draw[very thick] (c3.center) -- (c4.center);
			
			\draw[->,color=red,thick] (a3.center)++(-1,+1) -- +(1,0); 
			\draw[->,color=red,thick] (a3.center)++(-1,+1) -- +(135:1); 
			\draw[->,color=blue,thick] (a4.center)++(-1,+1) -- +(1,0); 
			\draw[->,color=blue,thick] (a4.center)++(-1,+1) -- +(135:1);
			
			\draw[->,color=red,thick] (-0.5,-0.5) -- +(1,0); 
			\draw[->,color=red,thick] (-0.5,-0.5) -- +(135:1); 
			\draw[->,color=blue,thick,dashed] (-0.5,-1) -- +(1,0); 
			\draw[->,color=blue,thick,dashed] (-0.5,-1) -- +(315:1); 
			
			\draw[->,color=red,thick] (-0.5,-0.5)++(11.5,0) -- +(1,0); 
			\draw[->,color=red,thick] (-0.5,-0.5)++(11.5,0) -- +(135:1); 
			\draw[->,color=blue,thick,dashed] (-0.5,-1)++(11.5,0) -- +(1,0); 
			\draw[->,color=blue,thick,dashed] (-0.5,-1)++(11.5,0) -- +(315:1);
			
			\draw (t1.center) -- +(6,0);
			\draw (t1.center)++(11.5,0) -- +(6,0);
			
			\draw[thick] (10,3.5) -- +(290:1.5);
			\draw[thick] (10,3.5) -- +(250:1.5) node[pos=0.66] (n1) {};
			\draw[thick] (n1.center) arc (250:290:1); 
			\draw[thick,->] (10,2.4) -- (10,1.2);
			
			\node[draw,single arrow, minimum height=20mm, minimum width=8mm, anchor=west] at (-7,0) {};
			\node[draw,single arrow, minimum height=20mm, minimum width=8mm, anchor=west] at (4.5,0) {};
			
			\node[alice, minimum size=1cm] at (3.5,-1.5) {};
			\node[alice, minimum size=1cm] at (12.5,1.5) {};
			
			\node[color=red] (l11) at (-13.7,1) {\large $ m_+ >0 $};
			\node[color=blue] (l12) at (-10.5,-2.5) {\large $ m_- <0 $};
			\node[right] (l13) at (-7.65,-0.75) {\large $m=0$};
			
			\node[color=red] (l21) at (-2,0.5) {\large $ m_+ >0 $};
			\node[color=blue] (l22) at (0,-2.1) {\large $ m_- <0 $};
			\node[right] (l23) at (4.8,-2.5) {\large $m=0$};
			
			\node[color=red] (l31) at (9.5,0.5) {\large $ m_+ >0 $};
			\node[color=blue] (l32) at (11.5,-2.1) {\large $ m_- <0 $};
			\node[above] (l33) at (15.3,-2.7) {\large $m=0$};
		\end{tikzpicture}
	\end{adjustbox}
	\caption{From left to right: position space before folding; position space after folding and Alice standing close to the equator; Alice after moving towards the \enquote{upper bulk}, now looking down at the lower rim.}
	\label{fig:Alice}
\end{figure}
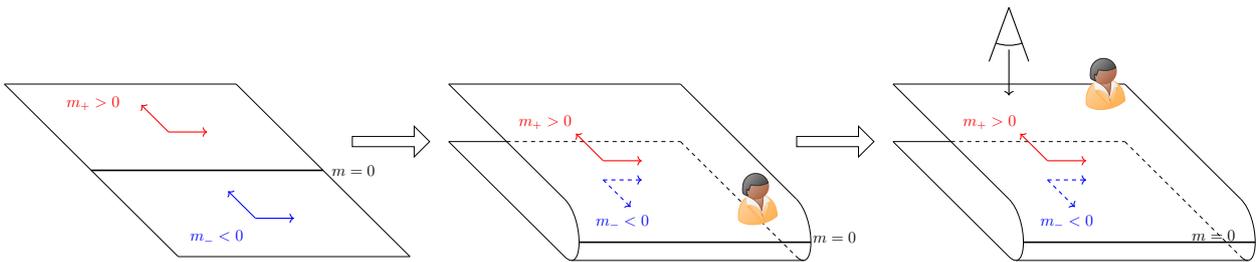

Aided by this intuition, one can formalize the bulk picture of system $H$. Let $ \bb{E} $ be the affine plane and $ \bb{E}_\pm $ two oppositely oriented copies of it, with natural identification $ i: \bb{E}_+ \to \bb{E}_- $ through $ \bb{E}_+ \equiv \bb{E} \equiv \bb{E}_- $. Let $ \bbR^2 $ be equipped with the canonical orientation. Let $ \varphi_\pm : \bb{E}_\pm \to \bbR^2 $ be two orientation preserving charts (linear w.l.o.g.). Then $ i: \bbR^2 \to \bbR^2 $ acts between charts with $ \det i = -1 $, for example as
\begin{equation}
	i: (x_1, x_2) \mapsto (x_1,-x_2) \,. \label{eq:InversionSpace}
\end{equation}
Differently put, $ i \circ \varphi_\pm = \varphi_\mp $.

Orientability of a manifold is one of the prerequisites for carrying a spin structure and hence a Dirac operator. $ \bb{E}_\pm $ can thus host the Dirac Hamiltonian $ H_\pm $ with mass $ m_\pm $ respectively. Moving to momentum space, let $ \bb{E}^* $ be the dual of $ \bb{E} $, and let the Brillouin zones $ \bb{E}_\pm^* $ inherit orientation from $ \bb{E}_{\pm} $. We wish to express $ H_\pm $ in the coordinates $ \ul{k} $ of $ p \in \bb{E}^* $ induced by $ \varphi_\pm $. Here and in the following, let $ \varphi_+ (p) = (k_1,k_2)$. By Eq.\,\eqref{eq:InversionSpace}
\begin{equation}
	i: (k_1, k_2) \mapsto (k_1,-k_2) \,. \label{eq:InversionMomentum}
\end{equation}
When embedded in $H$ as asymptotic Hamiltonians, the action of $ H_{\pm} $ on a state of momentum $p \in \bb{E}^* $ can only differ by the mass $ m_\pm $. In other words, if
\begin{equation}
	H_+ (p) \coloneqq H_+ (\varphi_+ (p)) = (k_1,k_2,m_+) \cdot \vec{\sigma} \,, \label{eq:HPlusP}
\end{equation}
then
\begin{equation}
	H_- (p) \coloneqq H_- (\varphi_- (p)) = (k_1,k_2,m_-) \cdot \vec{\sigma} \,. \label{eq:HMinusP}
\end{equation}
By Eq.\,\eqref{eq:InversionMomentum}, this in turn means
\begin{equation}
	H_{\pm} (\ul{k}) = (k_1, \pm k_2, m_\pm) \cdot \vec{\sigma}
\end{equation}
as maps $ \bbR^2 \to L(\cal{H}) $, where $ L(\cal{H}) $ denotes linear operators on the Hilbert space.

The bundles $ \cal{E}_\pm $ are completely determined by the projections $ P_\pm $, which read
\begin{equation}
	P_\pm (p) \coloneqq \frac{ \id + (\vec{e}_\pm (p) \cdot \vec{\sigma}) }{2} \,,
\end{equation}
where
\begin{equation}
	\vec{e}_\pm (p) = \frac{1}{\omega_\pm (p)} (k \cos \theta, k \sin \theta, m_\pm) \,,
\end{equation}
by Eqs.\,(\ref{eq:HPlusP},\ref{eq:HMinusP}) and having switched to polar coordinates $ \varphi_+ (p) = (k \cos \theta, k \sin \theta) $. We notice that
\begin{equation}
	\lim_{k \to \infty} \vec{e}_+ (p) = \lim_{k \to \infty} \vec{e}_- (p) \ \Longrightarrow \ \lim_{k \to \infty} P_+ (p) = \lim_{k \to \infty} P_- (p) \,, \label{eq:PEquality}
\end{equation}
and such limits are finite, albeit direction dependent. Given existence of these limits, just as the real line can be compactified to a closed interval, so can the planes $ \bb{E}^*, \bb{E}_\pm^* $ be compactified to closed disks $ \bb{D}, \bb{D}_\pm $, where points at their boundary represent infinity in the corresponding direction. In turn, the disk is topologically equivalent to a hemisphere.
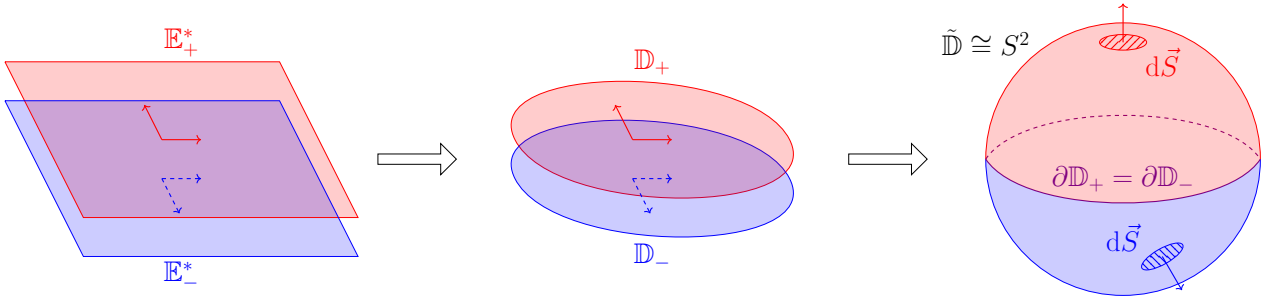
\begin{figure}[h]
	\centering
	\begin{adjustbox}{max size={.95\textwidth}{.8\textheight}}
		\begin{tikzpicture}
			\fill[red,opacity=0.2] (-16.5,2.5) -- (-9.5,2.5) -- (-7.5,-1.5) -- (-14.5,-1.5);
			\fill[blue,opacity=0.2] (-16.5,2.5-1) -- (-9.5,2.5-1) -- (-7.5,-1.5-1) -- (-14.5,-1.5-1);
			\draw[red] (-16.5,2.5) -- (-9.5,2.5) -- (-7.5,-1.5) -- (-14.5,-1.5) -- (-16.5,2.5);
			\draw[blue] (-16.5,2.5-1) -- (-9.5,2.5-1) -- (-7.5,-1.5-1) -- (-14.5,-1.5-1) -- (-16.5,2.5-1);
			
			\draw[->,color=red,thick] (-12.5,0.5) -- +(1,0); 
			\draw[->,color=red,thick] (-12.5,0.5) -- +(-0.447,0.894); 
			\draw[->,color=blue,thick,dashed] (-12.5,-0.5) -- +(1,0); 
			\draw[->,color=blue,thick,dashed] (-12.5,-0.5) -- +(0.447,-0.894);
			
			\node[color=red] (l11) at (-12,3) {\huge $ \bb{E}^*_+ $};
			\node[color=blue] (l12) at (-12,-3) {\huge $ \bb{E}^*_- $};
			
			\draw[red] (-3.5,0.5) .. controls (-2.5,-1.5) and (4.5,-1.5) .. (3.5,0.5) .. controls (2.5,2.5) and (-4.5,2.5) .. (-3.5,0.5); 
			\draw[blue] (-3.5,0.5-1) .. controls (-2.5,-1.5-1) and (4.5,-1.5-1) .. (3.5,0.5-1) .. controls (2.5,2.5-1) and (-4.5,2.5-1) .. (-3.5,0.5-1); 
			\fill[red,opacity=0.2] (-3.5,0.5) .. controls (-2.5,-1.5) and (4.5,-1.5) .. (3.5,0.5) .. controls (2.5,2.5) and (-4.5,2.5) .. (-3.5,0.5); 
			\fill[blue,opacity=0.2] (-3.5,0.5-1) .. controls (-2.5,-1.5-1) and (4.5,-1.5-1) .. (3.5,0.5-1) .. controls (2.5,2.5-1) and (-4.5,2.5-1) .. (-3.5,0.5-1); 
			
			\draw[->,color=red,thick] (-12.5+12,0.5) -- +(1,0); 
			\draw[->,color=red,thick] (-12.5+12,0.5) -- +(-0.447,0.894); 
			\draw[->,color=blue,thick,dashed] (-12.5+12,-0.5) -- +(1,0); 
			\draw[->,color=blue,thick,dashed] (-12.5+12,-0.5) -- +(0.447,-0.894);
			
			\node[color=red] (l21) at (0,2.5) {\huge $ \bb{D}_+ $};
			\node[color=blue] (l22) at (0,-2.5) {\huge $ \bb{D}_- $};
			
			\draw[red] (15.5,0) arc (0:180:3.5); 
			\draw[blue] (8.5,0) arc (180:360:3.5); 
			\draw[violet,thick] (8.5,0) .. controls (9.5,-1.5) and (14.5,-1.5) .. (15.5,0); 
			\draw[violet,dashed,thick] (8.5,0) .. controls (9.5,1.5) and (14.5,1.5) .. (15.5,0); 
			\fill[red,opacity=0.2] (15.5,0) arc (0:180:3.5) .. controls (9.5,-1.5) and (14.5,-1.5) .. (15.5,0); 
			\fill[blue,opacity=0.2] (8.5,0) arc (180:360:3.5) .. controls (14.5,-1.5) and (9.5,-1.5) .. (8.5,0); 
			
			\draw[red] (12,3) ellipse (0.6 and 0.2); 
			\draw[blue,rotate around={30:(13,-2.5)}] (13,-2.5) ellipse (0.6 and 0.2); 
			\draw[red,thick,->] (12,3)--(12,4);
			\draw[blue,thick,->] (13,-2.5) -- +(-60:1);
			\fill[pattern=north east lines, pattern color=red] (12,3) ellipse (0.6 and 0.2); 
			\fill[rotate around={30:(13,-2.5)}, pattern=north west lines, pattern color=blue] (13,-2.5) ellipse (0.6 and 0.2); 
			
			\node[color=black] (l31) at (8.5,3) {\huge $ \tilde{\bb{D}} \cong S^2 $};
			\node[color=red] (l32) at (13,2.5) {\huge $ \diff \vec{S} $};
			\node[color=blue] (l33) at (12,-2) {\huge $ \diff \vec{S} $};
			\node[color=violet] (l34) at (12,-0.5) {\huge $ \partial \bb{D}_+ = \partial \bb{D}_- $};
			
			\node[draw,single arrow, minimum height=20mm, minimum width=8mm, anchor=west] at (-7,0) {};
			\node[draw,single arrow, minimum height=20mm, minimum width=8mm, anchor=west] at (5,0) {};
		\end{tikzpicture}
	\end{adjustbox}
	\caption{Visualization of the compactification procedure. Left panel: oppositely oriented planar Brillouin zones $ \bb{E}^*_\pm $. Center panel: disks $\bb{D}_\pm$. Right panel: $ \bb{D}_\pm $ joined along the boundary $ \partial \bb{D}_+ = \partial \bb{D}_- $ to form an oriented sphere $ \tilde{\bb{D}} \cong S^2 $.}
	\label{fig:BulkIndex}
\end{figure}

Using the new Brillouin zones $ \bb{D}_\pm $, Eq.\,\eqref{eq:PEquality} is rewritten as
\begin{equation}
	P_+ (p) = P_- (p) \,, \qquad (p \in \partial \bb{D}) \,, \label{eq:DiskEdge}
\end{equation}
by $ \partial \bb{D}_+ = \partial \bb{D}_- $.

Note that 
\begin{equation}
	\tilde{\bb{D}} \coloneqq \bb{D}_+ \cup \bb{D}_-
\end{equation}
joined along $ \partial \bb{D}_+ = \partial \bb{D}_- $ is a sphere equipped with an orientation (see Figure \ref{fig:BulkIndex}), and in fact consistent on $ \bb{D}_\pm $.

By Eq.\,\eqref{eq:DiskEdge},
\begin{equation}
	P = P_+ \sqcup P_- \label{eq:GlobalP}
\end{equation}
defines a line bundle $ \cal{E} $ on $ \tilde{\bb{D}} $. We have thus been able to construct a bundle $ \cal{E} $ over compact base space $\tilde{\bb{D}} \simeq S^2$ associated with the global Hamiltonian $H$.
\begin{definition}
	The bulk index $ \cal{I} $ of the Hamiltonian $H$ defined in Eq.\,\eqref{eq:H} is
	\begin{equation}
		\cal{I} \coloneqq Ch (\cal{E}) \,.
	\end{equation}
\end{definition}
This index can now be computed. Formally
\begin{equation}
	Ch (\cal{E}) = \frac{1}{2 \pi} \int_{\tilde{\bb{D}} } \vec{e} (q) \cdot (\partial_{q_1} \vec{e} (q) \wedge \partial_{q_2} \vec{e}(q)) \diff q_1 \diff q_2 \,,
\end{equation}
where $ \vec{e} $ is the unit vector associated to $P$ and $ q = (q_1,q_2) $ some coordinates on the sphere. Operationally, it is natural to place a chart on each hemisphere and evaluate the integral accordingly:
\begin{align}
	Ch (\cal{E}) &= \frac{1}{2 \pi} \int_{R^2} \vec{e}_+ (k_1, k_2) \cdot (\partial_{k_1} \vec{e}_+ (k_1, k_2) \wedge \partial_{k_2} \vec{e}_+ (k_1,k_2)) \diff k_1 \diff k_2 \nonumber \\
	&- \frac{1}{2 \pi} \int_{R^2} \vec{e}_- (k_1, k_2) \cdot (\partial_{k_1} \vec{e}_- (k_1, k_2) \wedge \partial_{k_2} \vec{e}_- (k_1,k_2)) \diff k_1 \diff k_2 \nonumber \\
	&= C \hbar (\cal{E}_+) - C \hbar (\cal{E}_-) \,, \label{eq:TotChern}
\end{align}
where the second integral appears with a minus sign due to the negative orientation of $ \bb{E}^*_- $.

We have thus proven the following proposition.
\begin{proposition}
	The quantity $ C \hbar (\cal{E}_+) - C \hbar (\cal{E}_-) $ is a genuine topological invariant, namely there exists a bundle $ \cal{E} $ over compact base space $ \tilde{\bb{D}} \simeq S^2 $ such that
	\begin{equation}
		Ch (\cal{E}) = C \hbar (\cal{E}_+) - C \hbar (\cal{E}_-) \,. \label{eq:BulkInv}
	\end{equation}
\end{proposition}
\begin{remark}
	We once again stress that Eq.\,\eqref{eq:BulkInv} has already been rigorously proven in \cite{BalInterface}. The derivation above is nonetheless alternative, in that it requires little to no familiarity with topological invariants, and leverages a \enquote{pictorial} intuition that was not made explicit before.
\end{remark}
\begin{remark} \label{rem:Degree}
	Here is another equivalent way of tackling non-compact Brillouin zones. Instead of integrating the Berry curvature on the full Brillouin zone (thereby obtaining the Chern number), one could integrate the Berry connection along some closed contour. The choice of contour is arbitrary, unless its \enquote{radius} is sent to infinity. Equivalence with our approach is now apparent: by Stokes' Theorem, the line integral of the connection is equal to the surface integral of the curvature in the interior of the curve, which again consists in the entire Brillouin zone.
	
    In our setup, the stated equivalence can be formally phrased as
	\begin{equation}
		C \hbar (\cal{E}_\pm) = \frac{1}{2} \op{sgn} (m_\pm) = \frac{1}{2} \op{deg} (\vec{e}_\pm^{\,\infty} ) \,, 
	\end{equation}
	where 
	\begin{equation}
		\vec{e}_\pm^{\, \infty} (\theta) \coloneqq \lim_{k \to \infty} \vec{e}_\pm (k, \theta) \label{eq:ConnectionCurvature}
	\end{equation}
    shall be seen as maps $ \vec{e}_\pm^{\, \infty}: S^1 \ni \theta \to \partial \mathbb{D}_+ \equiv \partial \mathbb{D}_- $ from a positively oriented circle to the equator of $ \tilde{\mathbb{D}} $. Clearly, $ \op{deg} (\vec{e}_+^{\, \infty} ) = - \op{deg} (\vec{e}_-^{\, \infty} ) $ because the two have to sweep the equator in opposite directions to keep the interior of $ \mathbb{D}_+ $ or $ \mathbb{D}_- $ to their left, respectively. Given our sign conventions,
	\begin{equation}
		\op{deg} (\vec{e}_\pm^{\, \infty} ) = \op{sgn} (m_\pm) \,,
	\end{equation}
	whence Eq.\,\eqref{eq:ConnectionCurvature}.
	
    Even this alternative method demands gluing two Hamiltonians with opposite mass to obtain an integer bulk index.
\end{remark}

\section{Edge index and bulk-edge correspondence} \label{sec:edge}

Ever since \cite{Hatsugai93PRL,Hatsugai93PRB}, many topological insulators have been shown to exhibit bulk-edge correspondence. The latter means that the edge index, typically the signed number of chiral states propagating along the boundary, coincides with the bulk invariant. Our model Hamiltonian $H$ enjoys the same property. We prove this claim by defining an edge index $ \cal{I}^\# $, computing its value and verifying it matches $ \cal{I} $. The index $ \cal{I}^\# $ is moreover proven independent on the choice of mass profile $ m(x_2) $, as expected of a topological quantity.

The following definition of the edge index is a modification of the fiducial line approach (see Ref.\,\cite{GP13} and references therein), more suited for unbounded operators.
\begin{definition}
	Consider $H$ as in Eq.\,\eqref{eq:H}, and let $ \tilde{m} \coloneqq \min \{ |m_-|, |m_+| \} $. Introduce a Fermi line $ \mu_\epsilon (k_1) $, approximating the lower rim of the upper band from below:
	\begin{equation}
		\mu_\epsilon (k_1) \coloneqq -\epsilon + \sqrt{k_1^2 + \tilde{m}^2} \,,
	\end{equation}
	where $ \epsilon > 0 $ a (small) constant, to be chosen so that $ \mu_\epsilon $ lies entirely in the bulk gap. Denote by $ \psi_j (\ul{x}) \coloneqq \eul^{\im k_1 x_1} \hat{\psi}_j (x_2;k_1) $ the bound eigenstate with dispersion relation $ \omega_j (k_1) $,
	\begin{equation}
		H(k_1) \hat{\psi}_j (x_2;k_1) = \omega_j (k_1) \hat{\psi}_j (x_2;k_1) \,. \label{eq:TISEk1}
	\end{equation}
	Then, the edge index $ \cal{I}^\# $ is 
	\begin{equation}
		\cal{I}^\# \coloneqq - \sum_{j \in J} I (\mu_\epsilon, \omega_j) \,, \label{eq:EdgeIndex}
	\end{equation}
	with $ I (\mu_\epsilon, \omega_j) $ the intersection number of the Fermi line with the $ j $-th edge channel $ \omega_j $, and $J$ the index set of bound eigenstates. \label{def:EdgeIndex}
\end{definition}
\begin{remark}
	A word on conventions (see also Fig.\,\ref{fig:EdgeIndex}). Both $ \mu_\epsilon (k_1) $ and $ \omega_j (k_1), \, (j \in J) $ are curves on the $ (k_1,\omega) $-plane. The intersection number between the $ k_1 $ and $ \omega $ axes is taken to be $ +1 $. The minus sign in Eq.\,\eqref{eq:EdgeIndex} is then necessary to recover the standard condensed matter convention \cite{GP13,Hatsugai93PRB,Hatsugai93PRL} that states \textit{born} on the \textit{lower rim} of a bulk band (while proceeding in the positive $ k_1 $ direction) contribute $ +1 $ to the edge index. Our choice is also equivalent to counting left-movers (right-movers) positively (negatively).
\end{remark}

The edge index is computed by finding the discrete spectrum and the relative eigenstates. Systems like our $H$ have been extensively studied elsewhere, see Refs.\,\cite{HK10,JackiwRebbi}. Two eigenstates $ \psi^{R/L} (\ul{x}; k_1) $ (right- and left-moving respectively) are known for any profile of the mass. They read
\begin{equation}
	\psi^{R/L} (\ul{x}; k_1) = \eul^{\im k_1 x_1} \hat{\psi}^{R/L} (x_2;k_1) \,, \qquad  \hat{\psi}^{R/L} (x_2;k_1) = N_{R/L} \eul^{\pm \int_0^{x_2} m(x_2') \diff x_2' } \begin{pmatrix} 1 \\ \pm 1 \end{pmatrix} \,, \label{eq:BoundStates}
\end{equation}
with $N_{R/L} \in \bbC$ some normalization constants. Their dispersion relation is $ \omega_{R/L} (k_1) = \pm k_1 $, and they are bound for $ m(x_2) $ monotonically decreasing or increasing, respectively. Since
\begin{equation}
	I (\mu_\epsilon, \omega_{R/L}) = \pm 1 \,,
\end{equation}
their contribution to the edge index would be $ \mp 1 $, respectively. We claim that these are the only net contributions, namely all other bound states (if they exist) intersect the Fermi line an even number of times with opposite signs. Equivalently, they emerge from and disappear into the same bulk band.
\begin{proposition}
    Let $ \cal{I}^\# $ as in Def.\,\ref{def:EdgeIndex}, and let $ m(x_2) $ be monotonous: $ m'(x_2) \geq 0 $ or $ m'(x_2) \leq 0 $ for all $x_2$. Then
	\begin{equation}
		\cal{I}^\# = \op{sgn} (m') \,.
	\end{equation}
	The unique bound state giving a net non-zero contribution to $ \cal{I}^\# $ for $ m' \geq 0 $ ($ m' \leq 0 $) is $ \psi^{L} $ ($ \psi^R $), cf.\,Eq.\,\eqref{eq:BoundStates}. Its dispersion relation reads $ \omega_L = -k_1 $ ($\omega_R = +k_1$). \label{prop:EdgeIndex}
\end{proposition}
The full proof is reported in App.\,\ref{app:ProofEdgeIndex}, but its main elements can be sketched here. Let $ \psi = (u,v)^T $ be a candidate solution of the eigenvalue equation \eqref{eq:TISEk1}. The latter consists in two coupled first order ODEs. Decoupling them results in second order ODEs for $u$ and $v$. These are formally equivalent to time-independent Schr\"odinger equations, with an effective potential dependent on $ m $ and $ m' $. The bound state energies must lie above the infimum of this potential and below its asymptotic value. Such a condition confines possible edge channels to the region of the $ (k_1,\omega) $ plane between the \enquote{light cone} $ | \omega | = | k_1 | $ and $ \sigma_{\mathrm{e}} (H) $, see Figure \ref{fig:EdgeIndex}. If a channel $ \omega_i (k_1) $ connects the bands, thus giving a non-zero contribution to $ \cal{I}^\# $, it must cross the $ \omega = k_1 = 0 $ point. However, only two states solve $ H(0) \hat{\psi} (x_2,0) = 0 $, and they are the ones reported in Eq.\,\eqref{eq:BoundStates}. All other edge modes, if they exist, thus emerge from and disappear into the same band. Their net contribution to $ \cal{I}^\# $ is zero, whence the claim.   

For $ m' \geq 0 $, as originally assumed in Section \ref{sec:Setup}, Prop.\,\ref{prop:EdgeIndex} implies bulk-edge correspondence in the form of
\begin{equation}
	\cal{I} = +1 = \cal{I}^\# \,.
\end{equation}
\begin{remark}
	Our edge index $ \cal{I}^{\#} $ differs from $ 2 \pi \sigma_I $, the transport index defined in Eq. (4) of \cite{BalInterface}, by the choice of the Fermi line. The two would indeed coincide if we set $ \mu (k_1) = E_0 $ in Def.\,\ref{def:EdgeIndex}, for some constant $E_0$ in the global bulk gap $ [-\tilde{m},\tilde{m}] $. This difference is inconsequential for the Dirac model, but it will lead to a disagreement in the case of shallow-water waves (see Section \ref{sec:SW}).
\end{remark}
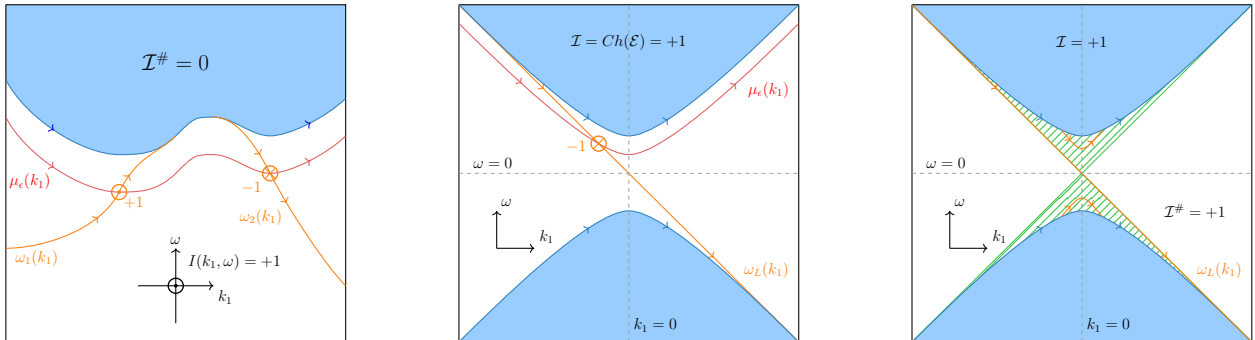
\begin{figure}[h]
	\centering
	\begin{adjustbox}{max size={.95\textwidth}{.8\textheight}}
		\pgfdeclarelayer{background}
		\pgfsetlayers{background,main}
		\begin{tikzpicture}
			\draw[thick] (-16.5,4.5) -- (-7.5,4.5) -- (-7.5,-4.5) -- (-16.5,-4.5) -- (-16.5,4.5); 
			\draw[thick] (-16.5+12,4.5) -- (-7.5+12,4.5) -- (-7.5+12,-4.5) -- (-16.5+12,-4.5) -- (-16.5+12,4.5); 
			\draw[thick] (-16.5+24,4.5) -- (-7.5+24,4.5) -- (-7.5+24,-4.5) -- (-16.5+24,-4.5) -- (-16.5+24,4.5); 
			
			\begin{scope}[decoration={markings, mark=at position 0.5 with {\arrow[blue,thick]{>}}}]
				\draw[blueBandEdge,postaction={decorate},thick] (-16.5,2.5) .. controls (-16,1.5) and (-14.5,0.5) .. (-13.5,0.5); 
				\draw[blueBandEdge,,thick] (-13.5,0.5) .. controls (-13,0.5) and (-12.5,0.5) .. (-12,1);
				\draw[blueBandEdge,thick] (-12,1) .. controls (-11.5,1.5) and (-11.5,1.5) .. (-11,1.5);
				\draw[blueBandEdge,thick] (-11,1.5) .. controls (-10.5,1.5) and (-10,1) .. (-9.5,1);
				\draw[blueBandEdge,postaction={decorate},thick] (-9.5,1) .. controls (-9,1) and (-8,1.5) .. (-7.5,2);
				
				\begin{pgfonlayer}{background}
					\fill[blueBand] (-16.5,4.5) -- (-16.5,2.5) .. controls (-16,1.5) and (-14.5,0.5) .. (-13.5,0.5) .. controls (-13,0.5) and (-12.5,0.5) .. (-12,1) .. controls (-11.5,1.5) and (-11.5,1.5) .. (-11,1.5) .. controls (-10.5,1.5) and (-10,1) .. (-9.5,1) (-9.5,1) .. controls (-9,1) and (-8,1.5) .. (-7.5,2) -- (-7.5,4.5) -- (-16.5,4.5);
				\end{pgfonlayer}
			\end{scope}
			
			\begin{scope}[decoration={markings, mark=at position 0.5 with {\arrow[FermiLine,thick]{>}}}]
				\draw[FermiLine,postaction={decorate},thick] (-16.5,2.5-1) .. controls (-16,1.5-1) and (-14.5,0.5-1) .. (-13.5,0.5-1) node[style=midway,red,anchor=north east] (n1) {\large $ \mu_\epsilon (k_1) $}; 
				\draw[FermiLine,thick] (-13.5,0.5-1) .. controls (-13,0.5-1) and (-12.5,0.5-1) .. (-12,1-1);
				\draw[FermiLine,thick] (-12,1-1) .. controls (-11.5,1.5-1) and (-11.5,1.5-1) .. (-11,1.5-1);
				\draw[FermiLine,thick] (-11,1.5-1) .. controls (-10.5,1.5-1) and (-10,1-1) .. (-9.5,1-1);
				\draw[FermiLine,postaction={decorate},thick] (-9.5,1-1) .. controls (-9,1-1) and (-8,1.5-1) .. (-7.5,2-1);
			\end{scope}
			
			\begin{scope}[decoration={markings, mark=at position 0.75 with {\arrow[EdgeState,thick]{>}}}]
				\draw[EdgeState,postaction={decorate},thick] (-16.5,-2) .. controls (-15.5,-2) and (-14,-1.5) .. (-13.5,-0.5) node[pos=0.25,orange,anchor=north] (n2) {\large $ \omega_1 (k_1) $}; 
				\draw[EdgeState,postaction={decorate},thick] (-11,1.5) .. controls (-10.5,1.5) and (-10,1) .. (-9.5,0); 
			\end{scope}
			\begin{scope}[decoration={markings, mark=at position 0.25 with {\arrow[EdgeState,thick]{>}}}]
				\draw[EdgeState,postaction={decorate},thick] (-13.5,-0.5) .. controls (-13,0.5) and (-12.5,0.5) .. (-12,1); 
				\draw[EdgeState,postaction={decorate},thick] (-9.5,0) .. controls (-9,-1) and (-8,-2.5) .. (-7.5,-3) node[pos=0.25,orange,anchor=north east] (n3) {\large $ \omega_2 (k_1) $}; 
			\end{scope}
			
			\draw[->,thick] (-10.5-1.5,-4) -- (-10.5-1.5,-2);
			\draw[->,thick] (-11.5-1.5,-3) -- (-9.5-1.5,-3);
			\node[anchor=south] (a1) at (-10.5-1.5,-2) {\large $ \omega $};
			\node[anchor=north west] (a2) at (-9.5-1.5,-3) {\large $ k_1 $};
			\node[anchor=south west] (l1) at (-10.3-1.5,-2.7) {\large $ I(k_1,\omega) = +1 $};
			\path (-12,-3) pic {vector out={line width=1pt,black/black/0.2cm}};
			
			\path (-13.5,-0.5) pic {vector out={line width=1pt,EdgeState/EdgeState/0.2cm}} (-9.5,0) pic {vector in={line width=1pt,EdgeState / line width=0.7pt,EdgeState/ 0.2cm}};
			\node[EdgeState,anchor=north west] (i1) at (-13.5,-0.5) {\large $ +1 $};
			\node[EdgeState,anchor=north east] (i2) at (-9.5-0.1,0-0.1) {\large $ -1 $};
			
			\node (l2) at (-12,3) {\Large $ \cal{I}^\# = 0 $};
			
			\begin{scope}[decoration={markings, mark=at position 0.25 with {\arrow[EdgeState,thick]{>}}, mark=at position 0.75 with {\arrow[EdgeState,thick]{>}}}]
				\draw[EdgeState,thick,postaction={decorate}] (-4.5,4.5) -- (4.5,-4.5) node[EdgeState,pos=0.825,anchor=south west] (les) {\large $ \omega_L (k_1) $};
			\end{scope}
			
			\begin{scope}[decoration={markings, mark=at position 0.4 with {\arrow[blueBandEdge,thick]{>}}, mark=at position 0.6 with {\arrow[blueBandEdge,thick]{>}}}]
				\draw[blueBandEdge,thick,postaction={decorate}] (-4.5,4.5) .. controls (-3.5,3.5) and (-1,1) .. (0,1) .. controls (1,1) and (3.5,3.5) .. (4.5,4.5); 
				
				\draw[blueBandEdge,thick,postaction={decorate}] (-4.5,-4.5) .. controls (-3.5,-3.5) and (-1,-1) .. (0,-1) .. controls (1,-1) and (3.5,-3.5) .. (4.5,-4.5); 
				
				\begin{pgfonlayer}{background}
					\fill[blueBand] (-4.5,4.5) .. controls (-3.5,3.5) and (-1,1) .. (0,1) .. controls (1,1) and (3.5,3.5) .. (4.5,4.5) -- (-4.5,4.5); 
					\fill[blueBand] (-4.5,-4.5) .. controls (-3.5,-3.5) and (-1,-1) .. (0,-1) .. controls (1,-1) and (3.5,-3.5) .. (4.5,-4.5) -- (-4.5,-4.5); 
				\end{pgfonlayer}
			\end{scope}
			
			\begin{scope}[decoration={markings, mark=at position 0.2 with {\arrow[FermiLine,thick]{>}}, mark=at position 0.8 with {\arrow[FermiLine,thick]{>}}}]
				\draw[FermiLine,thick,postaction={decorate}] (-4.5,4) .. controls (-3.5,3) and (-1,0.5) .. (0,0.5) node[pos=0.79] (ip2) {} .. controls (1,0.5) and (3.5,3) .. (4.5,4) node[pos=0.65,red,anchor=north west] (lf) {\large $ \mu_\epsilon (k_1) $}; 
			\end{scope}
			\path (ip2.center) pic {vector in={line width=1pt,EdgeState / line width=0.7pt,EdgeState/ 0.2cm}};
			\node[orange,anchor=north east] at (-1,1) {\large $-1$};
			
			\draw[nicerGray,opacity=1, thick, dashed] (-4.5,0) -- (4.5,0) node[pos=0.1,black,anchor=south] {\large $ \omega = 0 $};
			\draw[nicerGray,opacity=1, thick, dashed] (0,-4.5) -- (0,4.5) node[pos=0.05,black,anchor=west] {\large $ k_1 = 0 $};
			
			\draw[->,thick] (-3.5,-2) -- (-3.5,-1) node[pos=1,anchor=south west] {\large $ \omega $};
			\draw[->,thick] (-3.5,-2) -- (-2.5,-2) node[pos=1,anchor=south west] {\large $ k_1 $};
			
			\node at (0,3.5) {\large $ \cal{I} = Ch (\cal{E}) = +1$};
			
			\begin{scope}[decoration={markings, mark=at position 0.4 with {\arrow[blueBandEdge,thick]{>}}, mark=at position 0.6 with {\arrow[blueBandEdge,thick]{>}}}]
				\draw[blueBandEdge,thick,postaction={decorate}] (-4.5+12,4.5) .. controls (-3.5+12,3.5) and (-1+12,1) .. (0+12,1) node[pos=0.85] (es1S) {} .. controls (1+12,1) and (3.5+12,3.5) .. (4.5+12,4.5) node[pos=0.15] (es1f) {}; 
				\draw[blueBandEdge,thick,postaction={decorate}] (-4.5+12,-4.5) .. controls (-3.5+12,-3.5) and (-1+12,-1) .. (0+12,-1) node[pos=0.85] (es2S) {} .. controls (1+12,-1) and (3.5+12,-3.5) .. (4.5+12,-4.5) node[pos=0.15] (es2f) {}; 
				
				\begin{pgfonlayer}{background}
					\fill[blueBand] (-4.5+12,4.5) .. controls (-3.5+12,3.5) and (-1+12,1) .. (0+12,1) .. controls (1+12,1) and (3.5+12,3.5) .. (4.5+12,4.5) -- (-4.5+12,4.5); 
					\fill[blueBand] (-4.5+12,-4.5) .. controls (-3.5+12,-3.5) and (-1+12,-1) .. (0+12,-1) .. controls (1+12,-1) and (3.5+12,-3.5) .. (4.5+12,-4.5) -- (-4.5+12,-4.5); 
				\end{pgfonlayer}
			\end{scope}
			
			\draw[nicerGray,opacity=1, thick, dashed] (-4.5+12,0) -- (4.5+12,0) node[pos=0.1,black,anchor=south] {\large $ \omega = 0 $};
			\draw[nicerGray,opacity=1, thick, dashed] (0+12,-4.5) -- (0+12,4.5) node[pos=0.05,black,anchor=west] {\large $ k_1 = 0 $};
			
			\draw[->,thick] (-3.5+12,-2) -- (-3.5+12,-1) node[pos=1,anchor=south west] {\large $ \omega $};
			\draw[->,thick] (-3.5+12,-2) -- (-2.5+12,-2) node[pos=1,anchor=south west] {\large $ k_1 $};
			
			\begin{pgfonlayer}{background}
				\draw[niceGreen,thick] (7.5,4.5) -- (16.5,-4.5);
				\draw[niceGreen,thick] (7.5,-4.5) -- (16.5,4.5);
			
				\fill[pattern=north east lines, pattern color=niceGreen] (7.5,4.5) -- (12,0) -- (16.5,4.5) .. controls (15.5,3.5) and (13,1) .. (12,1) .. controls (11,1) and (8.5,3.5) .. (7.5,4.5); 
				\fill[pattern=north east lines, pattern color=niceGreen] (7.5,-4.5) -- (12,0) -- (16.5,-4.5) .. controls (15.5,-3.5) and (13,-1) .. (12,-1) .. controls (11,-1) and (8.5,-3.5) .. (7.5,-4.5); 
			\end{pgfonlayer}
			
			\node at (12,3.5) {\large $ \cal{I} = +1$};
			\node at (15,-1) {\large $ \cal{I}^\# = +1$};
			
			\begin{scope}[decoration={markings, mark=at position 0.25 with {\arrow[EdgeState,thick]{>}}, mark=at position 0.75 with {\arrow[EdgeState,thick]{>}}}]
				\draw[EdgeState,thick,postaction={decorate}] (-4.5+12,4.5) -- (4.5+12,-4.5) node[EdgeState,pos=0.825,anchor=south west] (les) {\large $ \omega_L (k_1) $};
				\draw[EdgeState,thick,postaction={decorate}] (es1S.center) .. controls (12,0.5) and (12,0.5) .. (es1f.center);
				\draw[EdgeState,thick,postaction={decorate}] (es2S.center) .. controls (12,-0.5) and (12,-0.5) .. (es2f.center);
			\end{scope}
		\end{tikzpicture}
	\end{adjustbox}
	\caption{Left panel: Illustration of intersection numbers and related conventions. In this example, $ I (\mu_\epsilon,\omega_1) = +1$ and $ I (\mu_\epsilon,\omega_2) = -1$. By Eq. \eqref{eq:EdgeIndex}, $ \cal{I}^\# = -I(\mu_\epsilon, \omega_1) -I (\mu_\epsilon, \omega_2) = -1 +1 = 0 $. Middle panel: contribution of $ \omega_L (k_1) $ to the edge index. Right panel: allowed region (\enquote{light cone}) for bound states. If there exist bound states besides the one with dispersion relation $ \omega_L $, their net contribution to $ \cal{I}^\#$ is zero.}
	\label{fig:EdgeIndex}
\end{figure}

\section{Application: Rotating shallow-water model} \label{sec:SW}

The techniques developed above (in particular, compactification of the Brillouin zone via \enquote{fermion doubling}) apply to the rotating shallow-water model \cite{DMV17,TDV20,GJT21}, where they lead to the proof of a new instance of bulk-edge correspondence. The shallow-water Hamiltonians $ H_{SW}^\pm $ are given as the spin-1 counterparts of $ H_\pm $, cf.\,\eqref{eq:HPlusMinus}. They have no well-defined bulk invariant, but acquire one when embedded (as asymptotic limits, in the strong sense) into $ H_{SW} $, a Hamiltonian with variable angular velocity $f$. $H_{SW}$ has associated Bloch bundle $ \cal{E}_{SW} $ over $ S^2 $, whose Chern number is the bulk index $ \cal{I} = Ch (\cal{E}_{SW}) $. The edge invariant $ \cal{I}^\# $ is defined in complete analogy with Def.\,\ref{def:EdgeIndex}. Its value is explicitly computed for a velocity profile $ f(x_2) = f \op{sgn} (x_2) $, $f>0$ constant, resulting in $ \cal{I} = \cal{I}^\# = +2 $.

Consider $H_\pm$ as in Eq.\,\eqref{eq:HPlusMinus}. Substitute $ m_\pm $ by $ f_\pm $ and $ \vec{\sigma} $ by $ \vec{S} = (S_1,S_2,S_3) $, where
\begin{equation}
	S_1 = 
	\begin{pmatrix}
		0 & 1 & 0 \\
		1 & 0 & 0 \\
		0 & 0 & 0 
	\end{pmatrix} \,, \qquad
	S_2 = 
	\begin{pmatrix}
		0 & 0 & 1 \\
		0 & 0 & 0 \\
		1 & 0 & 0 
	\end{pmatrix} \,, \qquad
	S_3 = 
	\begin{pmatrix}
		0 & 0 & 0 \\
		0 & 0 & - \im \\
		0 & +\im & 0
	\end{pmatrix} \,. \label{eq:SMatrices}
\end{equation}
This results in new spin-1 operators
\begin{equation}
	H_{SW}^\pm \coloneqq (p_1,p_2,f_\pm) \cdot \vec{S} \,. \label{eq:HpmSW}
\end{equation}
Eigenstates of $ H_{SW}^\pm $ are solutions of the shallow-water equations, used among other things to describe Earth's oceanic layers \cite{DMV17}. The coefficients $ f_\pm $ are thus recognized as positive (negative) angular velocities, determining sign and strength of the Coriolis force. On Earth, $ f=0 $ at the equator.

The Hamiltonians $ H_{SW}^\pm $ are translation invariant (by construction) and particle-hole symmetric. Writing them fiber-wise, one finds the spectra
\begin{equation}
	\sigma ( H_{SW}^\pm (\ul{k}) ) = \sigma_{\mathrm{e}} ( H_{SW}^\pm (\ul{k}) ) = \{ - \sqrt{k^2 + f_\pm^2} \} \cup \{0\} \cup \{ \sqrt{k^2 + f_\pm^2} \} \,, \quad (\ul{k} \in \bbR^2) \,. \label{eq:Bands}
\end{equation}
The Brillouin zone is again $ \bbR^2 $ (non-compact). As in Sec.\,\ref{sec:Setup}, define the bundles $ \cal{E}_{SW}^\pm $ associated with the positive band of $ H_{SW}^\pm (\ul{k}) $. These bundles have ill-defined Chern numbers. Even so, computing \eqref{eq:ChernIntegral} happens to return integer values
\begin{equation}
	C\hbar (\cal{E}^{\pm}_{SW}) = \pm 1 \,.
\end{equation} 

Introduce the Hamiltonian
\begin{equation}
	H_{SW} = (p_1, p_2, f) \cdot \vec{S} \,, \label{eq:HSW}
\end{equation}
where $f$ is a function of position operator $x_2$ and we assume: $ f(x_2) $ continuously differentiable and monotonous; $ f(x_2) \to f_\pm $ as $ x_2 \to \pm \infty $; $ f(0)=0 $. $ H_{SW} $ tends to $ H_{SW}^\pm $ asymptotically, in the sense that
\begin{equation}
	( H_{SW} - H_{SW}^\pm ) T_a \overset{s}{\longrightarrow} 0 \,, \qquad a \to \pm \infty \,. \label{eq:AsymptoticConvergence}
\end{equation}
The essential spectrum of the partially Fourier-transformed operator 
\begin{equation}
	H_{SW} (k_1) = (k_1, p_2, f) \cdot \vec{S} \label{eq:PartFourierHam}
\end{equation}
is
\begin{equation}
	\sigma_{\mathrm{e}} (H_{SW} (k_1)) = \left\{ \omega \in \bbR \ | \ |\omega| \geq \sqrt{k_1^2 + \min \{f_+^2, f_-^2\} } \right\} \cup \{ 0 \} \,. \label{eq:EssSpecHSW}
\end{equation}
Unlike Eq.\,\eqref{eq:EssSpecH}, the proof of \eqref{eq:EssSpecHSW} is non-trivial because of the flat zero-energy band. We defer this proof to App. \ref{App:EssSpec}.

The Brillouin zone of $ H_{SW} $ is compactified to the sphere $S^2$ just as in Sec.\,\ref{sec:Bulk}. The bundle $ \cal{E}_{SW} $ constructed over it has
\begin{equation}
	Ch (\cal{E}_{SW}) = C\hbar (\cal{E}^+_{SW}) - C\hbar (\cal{E}^-_{SW}) = 2 \,,
\end{equation}
which we take as our bulk index $ \cal{I} $.

Bulk-edge correspondence is finally probed in a special case, $ f_+ = f = -f_- $ and $ f(x_2) = f \op{sgn} (x_2) $. The bound states, found by direct computation (see App.\,\ref{App:InterfaceDetails}), are precisely two:
\begin{equation}
	\psi_{a} (\ul{x},k_1) = C_a \eul^{\im k_1 x_1} \eul^{- f |x_2|}
	\begin{pmatrix}
		1 \\
		-1 \\
		0
	\end{pmatrix} \,, \qquad 
	\psi_{b} (\ul{x}, k_1) = C_b \eul^{\im k_1 x_1} \eul^{- |k_1| |x_2|}
	\begin{pmatrix}
		0 \\
		\op{sgn} (x_2) \op{sgn} (k_1) \\
		\im
	\end{pmatrix} \,, \label{eq:EdgeStatesSW}
\end{equation}
with $ C_a, C_b \in \bbC $ normalization constants. The corresponding dispersion relations are
\begin{equation}
	\omega_a (k_1) = -k_1 \,, \qquad \omega_b (k_1) = f \op{sgn} (k_1) \,, \label{eq:EdgeChannels}
\end{equation}
as illustrated in Figure \ref{fig:EdgeIndexSW}.

Computing the edge index $ \cal{I}^\# $ as given in Def.\,\ref{def:EdgeIndex} yields
\begin{equation}
	\cal{I}^\# = - I (\mu_\epsilon, \omega_a) - I (\mu_\epsilon, \omega_b) = +2 = \cal{I} \,, \label{eq:BECSW}
\end{equation}
and Eq.\,\eqref{eq:BECSW} represents another instance of bulk-edge correspondence.
\begin{figure}[hbt]
	\centering
	\includegraphics[width=\linewidth]{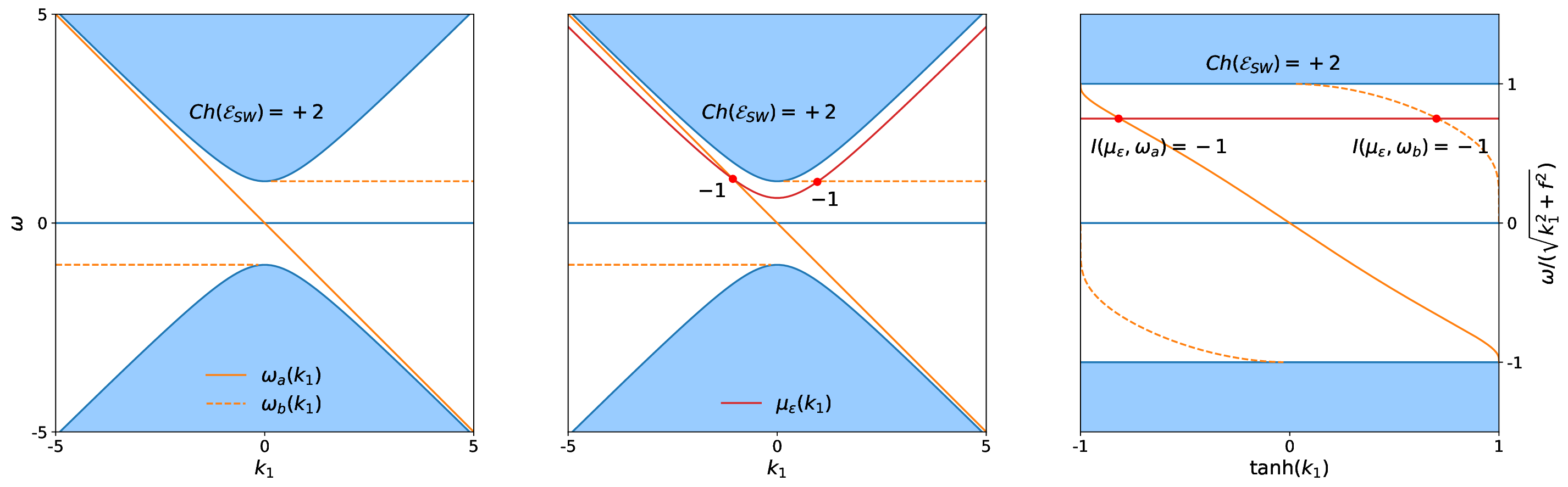}
	\caption{Edge spectrum and edge index when $ f(x_2) = f \op{sgn} (x_2) $. From left to right: Edge channels; Edge index as per Def.\,\ref{def:EdgeIndex}; Edge index in a rescaled picture $ \omega / \sqrt{k_1^2 + f^2} $ against $ \tanh (k_1) $.}
	\label{fig:EdgeIndexSW}
\end{figure}
\begin{remark} \label{rem:GeneralSolutions}
	It should be noticed that the solutions in Eq.\,\eqref{eq:BoundStates} of the Dirac eigenvalue problem extend immediately to the shallow-water setup. Indeed,
	\begin{equation}
		\psi^{R/L}_{SW} (\ul{x}; k_1) = \eul^{\im k_1 x_1} \hat{\psi}^{R/L}_{SW} (x_2;k_1) \,, \qquad  \hat{\psi}^{R/L}_{SW} (x_2;k_1) = C_{R/L} \eul^{\pm \int_0^{x_2} f(x_2') \diff x_2' } \begin{pmatrix} 1 \\ \pm 1 \\ 0 \end{pmatrix} \,, \qquad C_{R/L} \in \bbC \,, \label{eq:BoundStatesSW}
	\end{equation}
	are eigenstates of $ H_{SW} $ with eigenvalues $ \omega_{R/L} = \pm k_1 $ for any integrable angular velocity profile $ f(x_2) $. Moreover, $ \psi^R_{SW} $ ($\psi^L_{SW}$) is certainly bound for $ f(x_2) $ monotonically decreasing (increasing). If $ f(x_2) = f \op{sgn} (x_2) $, $ \psi^L_{SW} = \psi_a $ (cf.\,Eq.\,\eqref{eq:EdgeStatesSW}). 
	
	By contrast, $ \psi_b $ is a new solution with no counterpart in the Dirac case.
\end{remark}
\begin{remark}
	Ref.\,\cite{BalInterface} predicts $ 2 \pi \sigma_I = \op{sgn} (f) $ when $ f(y) = f \op{sgn} (y) $. Picking $ f >0 $ as above, the author thus concludes $ 2 \pi \sigma_I = +1 $, in disagreement with Eq.\,\eqref{eq:BECSW}. The discrepancy can be understood (at least) on two levels. On the heuristic level, $ \sigma_I $ physically represents a conductivity across some $ x_1 = \op{const} $ line. The state $ \psi_b $ of Eq.\,\eqref{eq:EdgeStatesSW} has zero group velocity in direction $x_1$, and does not contribute to transport. On the technical level, this index is only sensitive to edge channels that lie in the global bulk gaps $ (-f,0) $ and $ (0,f) $, i.e.\,not to $ \omega_b $ as in \eqref{eq:EdgeChannels}. This lowers the tally of edge channels by one.
\end{remark}

\section{Conclusions and future directions} \label{sec:Conclusion}

This work considered some of the topological properties of a 2D Dirac operator. In particular, the latter was shown to lack a bulk topological invariant in the constant mass case. It acquired one after being combined with its opposite-mass counterpart, as asymptotic limits of a global Hamiltonian with variable mass profile. It was useful to think of the latter as spatially interpolating between two distinct topological phases (positive and negative mass): The difference of their (ill-defined) invariants provided the global model with a (legitimate) bulk index $ \cal{I} $. This object was then related to the signed number $ \cal{I}^\# $ of bound states propagating along the interface $ x_2 = 0 $ of the two insulators. $ \cal{I}^\# $ was proven independent on the choice of mass profile $m(x_2)$, thus qualifying as an edge invariant. The equality $ \cal{I} = \cal{I}^\# $ represented an example of bulk-edge correspondence.

If some of the results on the 2D Dirac were known \cite{JackiwRebbi,HK10,BHW23,Bal22,BalInterface,QB23,Gruber06}, the real novelty of our treatment emerges when applying the same techniques to the shallow-water model. Our proof of bulk-edge correspondence, albeit only valid for the specific profile $ f(x_2) = f \op{sgn} (x_2) $ of the angular velocity, clashes with the claimed violations of \cite{GJT21} and \cite{BalInterface}. The first disagreement hints at a role of the odd-viscous regularizing term in spoiling the BEC. The second one suggests that computing transport indices w.r.t. constant fiducial lines in global bulk gaps may sometimes be inadequate when considering unbounded Hamiltonians.

From a more technical perspective, the results of App.\,\ref{App:EssSpec} on the essential spectrum of $ H_{SW} (k_1) $, cf.\,Eq.\,\eqref{eq:EssSpecHSW}, are also of interest. The authors are not aware of their existence elsewhere in the literature.

Future endeavours include: proving stability of the edge index under reasonable deformations of $f(x_2)$, in the shallow-water case; extending the \enquote{doubling} approach to different systems, starting perhaps with Dirac operators in higher spatial dimension; adapting the techniques to different symmetry classes; ultimately, achieving a general proof of bulk-boundary correspondence for models on the continuum. 

\section{Acknowledgements}

The authors are deeply grateful to Gian Michele Graf for his crucial contributions to Appendix \ref{App:EssSpec}, for numerous illuminating discussions and for useful advice on the structure and message of the paper. They moreover thank the anonymous referees for their punctual and constructive comments, which greatly helped improving the manuscript.

The first part of this work (Dirac Hamiltonian) was largely inspired by previous work of S.R., whereas A.T.\,developed the connection to shallow-water waves. The actual paper was conceptualized and structured by both. Most of the writing was then performed by A.T.\,, while S.R.\,contributed with App.\,\ref{app:BECFailureDirac}, some remarks and valuable revision work.

\begin{appendices}

\section{Proof of Proposition \ref{prop:EdgeIndex}} \label{app:ProofEdgeIndex}
The aim of this appendix is to flesh out the proof of Prop.\,\ref{prop:EdgeIndex}, so far only sketched in Sec.\,\ref{sec:edge}.
\begin{proof}
	Fix a fiber $ k_1 $ and look for solutions $ \hat{\psi} $ of
	\begin{equation}
		H(k_1) \hat{\psi} (x_2;k_1) = \omega (k_1) \hat{\psi} (x_2;k_1) \,, \label{eq:EgvlEqt}
	\end{equation}
	with energy $ \omega $ in the bulk gap. In the following, we drop the $ k_1 $-dependence since this parameter is fixed.
	
	Let $ \hat{\psi} (x_2) = (u (x_2),v (x_2))^T $. Eq.\,\eqref{eq:EgvlEqt} explicitly reads
	\begin{equation}
		\begin{pmatrix}
			m(x_2) & k_1 - \partial_2 \\
			k_1 + \partial_2 & -m(x_2)
		\end{pmatrix}
		\begin{pmatrix}
			u(x_2) \\
			v(x_2)
		\end{pmatrix} = \omega
		\begin{pmatrix}
			u(x_2) \\
			v(x_2)
		\end{pmatrix} \,, \label{eq:ExplicitEE}
	\end{equation}
	see Eq.\,\eqref{eq:HFiber}.
	
	Let $ s \coloneqq (u+v)/2 $, $ d \coloneqq (u-v)/2 $, or equivalently
	\begin{equation}
		u = s+d \quad \land \quad v = s-d \,. 
	\end{equation}
	Eq.\,\eqref{eq:ExplicitEE} gets rewritten as
	\begin{equation}
		\begin{cases}
			d'(x_2) + m(x_2) d(x_2) = (\omega - k_1) s(x_2) \\
			s'(x_2) - m(x_2) s(x_2) = -(\omega + k_1) d(x_2) \label{eq:SystemSD}
		\end{cases}
	\end{equation}
	in terms of $s,d$, where $ (\cdot)' = \partial_2 (\cdot) $. Eqs.\,\eqref{eq:SystemSD} can be decoupled at the price of going from first to second order ODEs
	\begin{equation}
		\begin{cases}
			\left( -\Delta + W_s (x_2) \right) s(x_2) = \omega^2 s (x_2) \,, & W_s (x_2) \coloneqq k_1^2 + m^2(x_2) + m'(x_2) \\
			\left( -\Delta + W_d (x_2) \right) d(x_2) = \omega^2 d (x_2) \,, & W_d (x_2) \coloneqq k_1^2 + m^2(x_2) - m'(x_2) \,. \label{eq:EffectiveSchro}
		\end{cases}
	\end{equation}
	Passing to second order may potentially introduce spurious solutions. Even so, if $s,d$ are (square-summable) solutions of System \eqref{eq:SystemSD}, then they are also (square-summable) solutions of System \eqref{eq:EffectiveSchro}. What follows only hinges on the contrapositive of the previous statement. 
	
	The ones in \eqref{eq:EffectiveSchro} are one-dimensional Schr\"odinger equations with potential $ W_{s/d} (x_2) $, bounded from below. $ \hat{\psi} $ is bound only if both $ u $ and $v$, or equivalently both $s$ and $d$, are. Bound eigenstates of a Schr\"odinger operator are known to have energy lying between the minimum and asymptotic value of the potential, see Eq.\,(2.91) in \,\cite{Griffiths} as a reference. By Eq.\,\eqref{eq:EffectiveSchro}, if $ s,d $ are both bound their energy $ \omega $ is such that
	\begin{equation}
		\max \{ \omega_s, \omega_d \} \leq \omega^2 < \min \{ k_1^2 + m_-^2, k_1^2 + m_+^2 \} \,, \label{eq:FirstBound}
	\end{equation}
	where
	\begin{equation}
		\omega_{s/d} \coloneqq \inf_{x_2} W_{s/d} (x_2) \,.
	\end{equation}
	
    If $ m' \geq 0 $ ($ m' \leq 0 $), then $ W_{s} (x_2) \geq W_d (x_2) $ ($ W_{s} (x_2) \geq W_d (x_2) $) for all $x_2$ and thus $ \max \{ \omega_s, \omega_d \} = \omega_s $ ($ \max \{ \omega_s, \omega_d \} = \omega_d $). Moreover, $ \omega_s \geq k_1^2 $ ($ \omega_d \geq k_1^2 $) by
	\begin{align}
		\omega_s &\coloneqq \inf_{x_2} (k_1^2 + m^2 (x_2) + m'(x_2)) \geq k_1^2 + \inf_{x_2} m^2 (x_2) + \inf_{x_2} m'(x_2) = k_1^2  \label{eq:OmegasBound} \\ 
		\left( \omega_d \right. &\coloneqq \inf_{x_2} (k_1^2 + m^2 (x_2) - m'(x_2)) \geq k_1^2 + \inf_{x_2} m^2 (x_2) + \sup_{x_2} m'(x_2) = \left. k_1^2 \right) \,. \label{eq:OmegadBound}
	\end{align}
	In either case, all bound states must have energy $ \omega $ satisfying
	\begin{equation}
		k_1^2 \leq \omega^2 < \min \{ k_1^2 + m_-^2, k_1^2 + m_+^2 \} \equiv k_1^2 + \tilde{m}^2 \,. \label{eq:ChainIneq}
	\end{equation}
	
	Notice that $ \omega^2 < k_1^2 + \tilde{m}^2 $ is nothing but the gap condition. Combined with $ k_1^2 \leq \omega^2 \leftrightarrow |\omega| \geq |k_1| $, it tells us that the allowed bound-state energies must lie between the light cone $ |\omega| = | k_1| $ and the essential spectrum $ \sigma_{\mathrm{e}} (H) $.
	
	As pointed out in Sec.\,\ref{sec:edge}, states emerging from and disappearing into the same band give a zero net contribution to $ \cal{I}^\# $. The relevant ones must thus cross the entire band gap. By Eq.\,\eqref{eq:ChainIneq}, they can only do so if their energy is $ \omega = 0 $ at $ k_1 = 0 $. At $ \omega = k_1 = 0 $, system \eqref{eq:SystemSD} reduces to
	\begin{equation}
		\begin{cases}
			d'(x_2) + m(x_2) d(x_2) = 0 \\
			s'(x_2) - m(x_2) s(x_2) = 0
		\end{cases}
		\quad \longleftrightarrow \quad 
		\begin{cases}
			d'(x_2) = D \eul^{- \int_0^{x_2} m(x_2') \diff x_2' } \\
			s'(x_2) = S \eul^{ \int_0^{x_2} m(x_2') \diff x_2' } \,,
		\end{cases}
	\end{equation}
	with $ S,D \in \bbC $ some normalization constants. When $ m' \geq 0 $, both are square-summable only for $ S=0 $, i.e.\,$ s(x_2) \equiv 0 $. Then $ u(x_2) = d(x_2) = -v(x_2) $, and
	\begin{equation}
		\hat{\psi} (x_2) = D \eul^{- \int_0^{x_2} m(x_2') \diff x_2' } 
		\begin{pmatrix}
			1 \\
			-1
		\end{pmatrix} \equiv \hat{\psi}^L (x_2)
	\end{equation}
	as claimed. Acting with the Hamiltonian reveals the dispersion relation $ \omega_L = -k_1 $, and $ \cal{I}^\# = - I (\mu_\epsilon, \omega_L) = +1 = \op{sgn} (m') $.
	
	Similarly, when $ m' \leq 0 $ square-summability requires $ d(x_2) \equiv 0 $, so that $ u(x_2) = s(x_2) = v(x_2) $ and
	\begin{equation}
		\hat{\psi} (x_2) = S \eul^{ \int_0^{x_2} m(x_2') \diff x_2' } 
		\begin{pmatrix}
			1 \\
			1
		\end{pmatrix} \equiv \hat{\psi}^R (x_2) \,.
	\end{equation}
	This time $ \omega_R = k_1 $, and $ \cal{I}^\# = - I (\mu_\epsilon, \omega_R) = -1 = \op{sgn} (m') $.
\end{proof}

\section{Essential spectrum of the shallow-water Hamiltonian} \label{App:EssSpec}

The aim of this appendix is to prove the following theorem, which implies Eq.\,\eqref{eq:EssSpecHSW} directly, and all the results that lead to it.
\vspace{1\baselineskip}
\begin{theorem} \label{thm:EssSpec}
	Let $ H_{SW} (k_1) $  be as in Eq.\,\eqref{eq:PartFourierHam} and $ H^\pm_{SW} (k_1) = (k_1, - \im \partial_2, f_\pm) \cdot \vec{S} $, where $ f_\pm = \lim_{x_2 \to \pm \infty} f(x_2) $. Then
	\begin{equation}
		\sigma_{\mathrm{e}} (H_{SW} (k_1)) = \bigcup\limits_{s = \pm} \sigma_{\mathrm{e}} ( H^s_{SW} (k_1) ) \label{eq:SpecEquality}
	\end{equation} 
	for all $ k_1 \in \bbR $.
\end{theorem}
The entire section concerns the shallow-water model, and in particular the operators $H_{SW}$ (cf.\,Eq.\,\eqref{eq:HSW}) and $ H^\pm_{SW} $ (cf.\,Eq.\,\eqref{eq:HpmSW}). Below, we thus drop subscripts $ (\cdot)_{SW} $ for readability.

Notice that Thm.\,\ref{thm:EssSpec} could also be stated as $ \sigma_{\mathrm{e}} (H (k_1)) = \cup_{s = \pm} \sigma (H^s (k_1)) $, because the spectrum of $ H^\pm $ is purely essential. This fact is proven as a warm-up and to display some of the techniques employed later.
\begin{lemma} \label{lem:EssSpecHpm}
	$$ \sigma (H^s (k_1)) = \sigma_{\mathrm{e}} ( H^s (k_1) ) \,, \qquad ( s = \pm ) \,. $$
\end{lemma}
\begin{proof}
	We recall Weyl's criterion: Let $A$ be a self-adjoint operator on some Hilbert space $\cal{H}$. Then $ \lambda \in \bbR $ belongs to $ \sigma_{\mathrm{e}} (A) $ if and only if there exists a so-called Weyl sequence $ \{ \psi_n \}_n \in \cal{H} $ such that $ \Vert \psi_n \Vert = 1 $, $ \psi_n \overset{w}{\to} 0 $ and
	\begin{equation}
		(A - \lambda) \psi_n \to 0 
	\end{equation}
	in Hilbert-space norm.
	
	Both operators $ H^s $ are translation invariant in $ x_2 $, i.e.\,commuting with $ T_a = \eul^{- \im p_2 a} $, $ (a \in \bbR) $. Then, any sequence of approximate eigenvectors, $ (H^s - \lambda ) \psi_n \to 0 $, $ (\Vert \psi_n \Vert = 1) $ can be turned into one, $ \tilde{\psi}_n \coloneqq T_{a_n} \psi_n $ that also has $ \tilde{\psi}_n \overset{w}{\to} 0 $ by suitable choice of $ a_n $. Thus, by Weyl's criterion $ \lambda \in \sigma (H^s (k_1)) $ implies $ \lambda \in \sigma_{\mathrm{e}} ( H^s (k_1) ) $ for all $k_1$.
	
	We actually have a slightly stronger statement: $ \chi \tilde{\psi}_n \to 0 $ for any $ \chi = \chi (x_2) $ with $ \chi (x_2) $ vanishing at $ x_2 \to + \infty $ (or $ x_2 \to - \infty $). This follows from
	\begin{equation}
		\chi T_a \overset{s}{\longrightarrow} 0 \,, \qquad (a \to + \infty)
	\end{equation}
	(respectively $ a \to - \infty $).
\end{proof}
The proof of the main theorem rests on the following lemma, which will be shown later.
\begin{lemma} \label{lem:KCompact}
	Let 
	\begin{equation}
		K = \chi H (H-z)^{-2} \,,
	\end{equation}
	where $\chi = \chi (x_2)$, $ \op{supp} \, \chi $ compact and we recall $ H \equiv H_{SW} $. Then $K$ is compact for all $ z \in \bbC, \op{Im} \, z \neq 0 $.
\end{lemma}
\begin{proof}[Proof of Thm.\,\ref{thm:EssSpec}]
	We start by showing
	\begin{equation}
		\sigma_{\mathrm{e}} (H^\pm (k_1)) \subseteq \sigma_{\mathrm{e}} (H (k_1)) \,, \qquad \forall k_1 \,. \label{eq:FirstInclusion}
	\end{equation}
	Assume $ \lambda \in \sigma_{\mathrm{e}} (H^+ (k_1)) $. Just as in the proof of Lemma \ref{lem:EssSpecHpm}, given a sequence $ \psi_n $ of approximate eigenstates we construct a new one 
	\begin{equation}
		\varphi_n \coloneqq T_{a_n} \psi_n \,.
	\end{equation}
	By suitable choice of $ a_n \in \bbR $, $ \varphi_n $ can be made weakly convergent to zero and such that $ \chi \varphi_n \to 0 $ for any $ \chi = \chi(x_2) $ vanishing at $x_2 \to + \infty$. This follows from
	\begin{equation}
		\chi \ T_a \overset{s}{\longrightarrow} 0 \,, \quad (a \to + \infty) \,. \label{eq:ChiT}
	\end{equation}
	By the same reason \eqref{eq:ChiT}, we have
	\begin{equation}
		\left( H (k_1) - H_+ (k_1) \right) T_a \overset{s}{\longrightarrow} 0 \,, \quad (a \to + \infty) \,,
	\end{equation}
	and likewise for $ H_- $ and $ a \to - \infty $. The desired inclusions, cf.\,Eq.\,\eqref{eq:FirstInclusion}, follow by the triangle inequality.
	
	The result just derived and Eq.\,\eqref{eq:Bands} imply $ 0 \in \sigma_{\mathrm{e}} (H (k_1)) \cap \sigma_{\mathrm{e}} (H^\pm (k_1)) $. We are thus left to prove
	\begin{equation}
		\sigma_{\mathrm{e}} (H (k_1)) \setminus \{0\} \subset \bigcup\limits_{s = \pm} \sigma_{\mathrm{e}} (H^s (k_1)) \,. \label{eq:DiffInclusion}
	\end{equation}
	Let $ \lambda \in \sigma_{\mathrm{e}} (H (k_1)) \setminus \{0\} $, and let $ g = g(\lambda') $ be such that: $ g(\lambda') = 1 $ for all $ \lambda' $ in some neighbourhood $ U \ni \lambda$; $ \op{supp} \, g $ compact; $ 0 \notin \op{supp} \, g $. Let $ \psi_n $ be a Weyl sequence for $\lambda$,
	\begin{equation}
		(H- \lambda) \psi_n \to 0 \,, \qquad \psi_n \overset{w}{\longrightarrow} 0 \,, \qquad \Vert \psi_n \Vert \to 1 \,.
	\end{equation}
	We claim
	\begin{equation}
		\chi \psi_n \to 0 \label{eq:ChiPsi}
	\end{equation}
	for any $ \chi = \chi (x_2) $ with $ \op{supp} \chi $ compact. Indeed, if this is the case $ \psi_n $ must escape towards either positive or negative spatial infinity, where $ H (k_1) $ tends to $ H^+ (k_1) $ or $ H^- (k_1) $ strongly. In the first (second) case, $ \lambda $ is in the essential spectrum of $ H^+ (k_1) $ ($ H^- (k_1) $) by the reasoning employed to prove Eq.\,\eqref{eq:FirstInclusion}.
	
	To show Eq.\,\eqref{eq:ChiPsi} we write
	\begin{equation}
		\chi \psi_n = \chi (1 - g(H)) \psi_n + \chi g(H) \psi_n \,. \label{eq:TwoSummands}
	\end{equation}
	Here and onwards, results are still meant for all $k_1$ despite omitting $ k_1 $ from the notation. Since $ (1 - g(H)) (H - \lambda)^{-1} $ is bounded, the first term in Eq.\,\eqref{eq:TwoSummands} is bounded by a constant times $ \Vert (H-\lambda) \psi_n \Vert $, and thus complies with Eq.\,\eqref{eq:ChiPsi} by hypothesis $ (H-\lambda) \psi_n \to 0 $. Since $ B \coloneqq (H-z)^2 H^{-1} g(H) $ is also bounded, we just need to show that the second term vanishes:
	\begin{equation}
		\chi H (H-z)^{-2} B \psi_n \to 0 \,.
	\end{equation}
	This follows from Lemma \ref{lem:KCompact}.
\end{proof}
Before proving Lemma \ref{lem:KCompact}, we state a further one, to be proven later.
\begin{lemma} \label{lem:Handp}
	For any $ z \in \bbC $ with $ \op{Im} (z) \neq 0 $, there exists $ C = C(z) < \infty $ such that
	\begin{equation}
		H (p_2^2 +1) H \leq C \left( (H-\bar{z}) (H - z) \right)^2 \,, \label{eq:Handp}
	\end{equation}
	where again $ H = H_{SW} $.
\end{lemma}
\begin{proof}[Proof of Lemma \ref{lem:KCompact}]
	By Lemma \ref{lem:Handp}, the operator
	\begin{equation}
		A \coloneqq (p_2 + \im) H (H-z)^{-2} \,, \label{eq:A}
	\end{equation}
	is such that $ A^* A \leq C $ for some finite $C$, i.e.\,it is bounded.
	
	Then, $ K = \chi (p_2+ \im)^{-1} A $ is compact if $ \chi (p_2 + \im)^{-1} $ is. This is true because $ \chi $ and $ (p_2 + \im)^{-1} $ vanish at infinity in position and momentum space respectively, which is a known sufficient condition for compactness (see Remark \ref{rem:Compact} below). The result follows.
\end{proof}
\begin{remark} \label{rem:Compact}
	Consider the operator $ T = f(x) g(p) $ on some functional Hilbert space over $\bbR^n$, say $ L^2 (\bbR^n) $ for definiteness, where $x$ and $p$ denote the position and momentum operator respectively. It is proven e.g.\, in \cite{RSIII} (page 47, Thm.\,XI.20) that, if $ f,g \in L^q (\bbR^n) $ as functions, then $ T \in J_q \subset K $, where $ J_q $ denotes the $q$-th Schatten class and $K$ the compact operators. Even if $ f $ is only vanishing at infinity, its restriction $ f_L(x) \coloneqq f(x) \chi_{[-L,L]} (x) $ is in $ L^2 (\bbR^n) $, with $ \chi_{[-L,L]} $ indicator function of the hypercube $ [-L,L]^{\times n} $. The same holds for $g$. Then, by the aforementioned result, $ T_L \coloneqq f_L (x) g_L (p) \in J_2 \subset K $. The claim that $T$ is compact now follows by $ T_L \to T \ (L \to \infty) $ in operator norm, and the fact that convergent sequences of compact operators have compact limit.
\end{remark}
\begin{remark}
	The boundedness of $A$ in Eq.\,\eqref{eq:A} informally states that, if $p_2$ diverges along some sequence of states, then $ H(k_1) $ has to do so too, or to tend to zero.
\end{remark}
\vspace{1\baselineskip}
Before proving Lemma \ref{lem:Handp}, we state yet another one, to be once again proven later. It is an identity, somewhat close in spirit to the Weitzenb\"ock formula \cite{Semmelmann10}.
\vspace{1\baselineskip}
\begin{lemma} \label{lem:AlgIdentity}
	Let $ \vec{d} = (d_1,d_2,d_3) $ with $ d_i = d_i^* $ self-adjoint operators on some Hilbert space $ \cal{H} $. Let $ \vec{S} = (S_1,S_2,S_3) $ be as in Eq.\,\eqref{eq:SMatrices}. Then, on $ \cal{H} \otimes \bbC^3 $, the following identity holds
	\begin{equation}
		( \vec{d} \cdot \vec{S} ) d^2 ( \vec{d} \cdot \vec{S} ) = ( \vec{d} \cdot \vec{S} )^4 + \frac{1}{2} \left( (\vec{d} \cdot \vec{S}) D + D^* (\vec{d} \cdot \vec{S}) \right) \,, \label{eq:AlgId1}
	\end{equation} 
	where $ d^2 = d_1^2 + d_2^2 + d_3^2 $ and
	\begin{equation}
		D = 
		\begin{pmatrix}
			\im (d_1 d_3 d_2 - d_2 d_3 d_1 ) & [d_3^2,d_1] & [d_3^2,d_2] \\
			[d_2^2, d_1] & \im ( d_3 d_2 d_1 - d_1 d_2 d_3 ) & - \im [d_2^2, d_3] \\
			[d_1^2, d_2] & - \im [d_1^2, d_3] & \im ( d_2 d_1 d_3 - d_3 d_1 d_2 )
		\end{pmatrix} \,. \label{eq:D}
	\end{equation}
\end{lemma}
\begin{remark}
	If the $ d_i $'s are numbers (or commuting operators), then $ D=0 $ and Eq.\,\eqref{eq:AlgId1} becomes trivial. In fact, the operator $ \vec{d} \cdot \vec{S} $ then has eigenvalues $ 0, \pm \Vert \vec{d} \Vert $.
\end{remark}
\begin{proof}[Proof of Lemma \ref{lem:Handp}]
	Apply the results of Lemma \ref{lem:AlgIdentity} to $ H \equiv H(k_1) = (k_1, p_2, f(x_2)) \cdot \vec{S} $ (operator on $ L^2 (\bbR) \otimes \bbC^3 $), having momentarily reinstated the momentum label $ k_1 $.
	
	Since $ \vec{d} = (k_1, - \im \partial_2, f) $ (where $ f \equiv f(x_2) $ here and throughout the proof)
	\begin{equation}
		D =
		\begin{pmatrix}
			-k_1 f' & 0 & 2 \im f f' \\
			0 & - k_1 f' & \im f'' - 2 f' p_2 \\
			0 & 0 & k_1 f' 
		\end{pmatrix} \,,
	\end{equation}
	with $ (\cdot)' = \partial_2 (\cdot) $. Moreover
	\begin{equation}
		H d^2 H = H^4 + \frac{1}{2} (H D + D^* H) \,.
	\end{equation}
	Now use $ A + A^* \leq \epsilon A A^* + \epsilon^{-1} $ for all $ \epsilon >0 $ and $ A = HD $ to bound
	\begin{equation}
		H D + D^* H \leq \epsilon H D D^* H + \epsilon^{-1} \,. \label{eq:EpsilonBound}
	\end{equation}
	Split $ D = D_1 + D_2 $ with
	\begin{equation}
		D_1 = 
		\begin{pmatrix}
			0 & 0 & 0 \\
			0 & 0 & - 2 f' p_2 \\
			0 & 0 & 0
		\end{pmatrix}
	\end{equation}
	and $ D_2 $ bounded. Then
	\begin{equation}
		D D^* \leq 2 (D_1 D_1^* + D_2 D_2^*) \,,
	\end{equation}
	which by 
	\begin{equation}
		D_1 D_1^* \leq 4 p_2 f' p_2 \leq C_1 p_2^2 \,, \qquad D_2 D_2^* \leq C_2
	\end{equation}
	becomes
	\begin{equation}
		D D^* \leq 2 ( C_1 p_2^2 + C_2 ) \,, \label{eq:DDstarBound}
	\end{equation}
	where $ C_1,C_2 > 0 $ some constants.
	
	Use $ p_2^2 \leq d^2 $ and plug Eq.\,\eqref{eq:DDstarBound} into Eq.\,\eqref{eq:EpsilonBound}:
	\begin{equation}
		H p_2^2 H \leq H d^2 H \leq H^4 + \epsilon H ( C_1 p_2^2 + C_2) H + (1/2 \epsilon) \,,
	\end{equation} 
	whence
	\begin{equation}
		(1- C_1 \epsilon) H p_2^2 H \leq H^4 + \epsilon C_2 H^2 + (1/2 \epsilon) \,. \label{eq:FinalBound}
	\end{equation}
	By picking (once and for all) $ \epsilon $ so small that $ C_1 \epsilon < 1/2 $, Eq.\,\eqref{eq:FinalBound} is recast as
	\begin{equation}
		H (p_2^2 + 1) H \leq 2 H^4 + ( 2 \epsilon C_2 + 1) H^2 + \epsilon^{-1} \,. \label{eq:SimilarBound}
	\end{equation}
	The thesis of the Lemma then follows by two observations. First, $ \forall z $ with $ \op{Im} z \neq 0 $, there exist $ \alpha, \beta > 0 $ such that
	\begin{equation}
		(H- \bar{z}) (H-z) \geq \alpha H^2 + \beta \,.
	\end{equation}
	Second, one can always find a finite, positive constant $ C = C(z) $ for which
	\begin{equation}
		2 H^4 + ( 2 \epsilon C_2 + 1) H^2 + \epsilon^{-1} \leq C (\alpha H^2 + \beta)^2 \,.
	\end{equation}
\end{proof}
\begin{proof}[Proof of Lemma \ref{lem:AlgIdentity}]
	Everything proceeds by direct calculation. We start by showing
	\begin{equation}
		d^2 ( \vec{d} \cdot \vec{S} ) = ( \vec{d} \cdot \vec{S} )^3 + D \,, \label{eq:d2dS}
	\end{equation}
	with $ D $ as in Eq.\,\eqref{eq:D}. This is seen by
	\begin{gather}
		( \vec{d} \cdot \vec{S} )^3 =
		\begin{pmatrix}
			- \im (d_1 d_3 d_2 - d_2 d_3 d_1) & d^2 d_1 - [d_3^2, d_1] & d^2 d_2 - [d_3^2, d_2] \\
			d^2 d_1 -[d_2^2, d_1] & -\im (d_3 d_2 d_1 - d_1 d_2 d_3) & - \im d^2 d_3 + \im [d_2^2, d_3] \\
			d^2 d_2 -[d_1^2, d_2] & + \im d^2 d_3 + \im [d_1^2, d_3] & -\im ( d_2 d_1 d_3 - d_3 d_1 d_2 ) 
		\end{pmatrix} \label{eq:dSCube} \\
		d^2 ( \vec{d} \cdot \vec{S} ) = 
		\begin{pmatrix}
			0 & d^2 d_1 & d^2 d_2 \\
			d^2 d_1 & 0 & - \im d^2 d_3 \\
			d^2 d_2 & + \im d^2 d_3 & 0  
		\end{pmatrix} \label{eq:dSd2dS}
	\end{gather}
	and comparison between Eqs.\,(\ref{eq:dSCube},\ref{eq:dSd2dS}).
	
	The conjugate of Eq.\,\eqref{eq:d2dS} reads
	\begin{equation}
		( \vec{d} \cdot \vec{S} ) d^2 = ( \vec{d} \cdot \vec{S} )^3 + D^* \,. \label{eq:dSd2}
	\end{equation}
	Thus
	\begin{equation}
		( \vec{d} \cdot \vec{S} ) d^2 ( \vec{d} \cdot \vec{S} ) = ( \vec{d} \cdot \vec{S} ) \left( ( \vec{d} \cdot \vec{S} )^3 + D \right) = ( \vec{d} \cdot \vec{S} )^4 + ( \vec{d} \cdot \vec{S} ) D \,, \label{eq:Summand1}
	\end{equation}
	having used Eq.\,\eqref{eq:d2dS}. However, its conjugate Eq.\,\eqref{eq:dSd2} implies
	\begin{equation}
		( \vec{d} \cdot \vec{S} ) d^2 ( \vec{d} \cdot \vec{S} ) = \left( ( \vec{d} \cdot \vec{S} )^3 + D^* \right)( \vec{d} \cdot \vec{S} ) = ( \vec{d} \cdot \vec{S} )^4 + D^* ( \vec{d} \cdot \vec{S} ) \,, \label{eq:Summand2}
	\end{equation}
	and the final claim follows from summing Eqs.\,(\ref{eq:Summand1},\ref{eq:Summand2}).
\end{proof}

\section{Interface states of the shallow-water model, complete calculation}\label{App:InterfaceDetails}

The following paragraphs are meant to derive the edge states and dispersion relations of Eqs.\,(\ref{eq:EdgeStatesSW},\ref{eq:EdgeChannels}), see Sec.\,\ref{sec:SW}. Such states are square-summable solutions of
\begin{equation}
	H_{SW} \psi = \omega \psi \,, \label{eq:EigenvalueEq}
\end{equation}
where $ H_{SW} $ is the shallow water Hamiltonian of Eq.\,\eqref{eq:HSW}
\begin{equation}
	H_{SW} = (p_1, p_2, f(x_2)) \cdot \vec{S} = (- \im \partial_1, - \im \partial_2, f(x_2)) \cdot \vec{S} \,, \label{eq:HSWApp}
\end{equation}
with $ f(x_2) = f \op{sgn} (x_2) $ ($ f > 0 $ constant) as our choice of angular velocity profile. In other words, and dropping the $ (\cdot)_{SW} $ subscripts here and below,
\begin{equation}
	H \rvert_{x_2 \gtrless 0} = H^\pm = (- \im \partial_1, -\im \partial_2, \pm f) \cdot \vec{S} \,.
\end{equation}
Notice that $ H^\pm $ are as in Eq.\,\eqref{eq:HpmSW}.

The Hamiltonian \eqref{eq:HSWApp} is translation invariant in direction $x_1$. Solutions of Eq.\,\eqref{eq:EigenvalueEq} must be of the form 
\begin{equation}
	\psi (\ul{x};k_1) = \eul^{\im k_1 x_1} \hat{\psi} (x_2;k_1) \,,
\end{equation}
and the eigenvalue problem can be rewritten fiber-wise as
\begin{equation}
	H (k_1) \hat{\psi} (x_2;k_1) = \omega (k_1) \hat{\psi} (x_2; k_1) \,, \qquad H(k_1) = (k_1, - \im \partial_2, f \op{sgn} (x_2)) \cdot \vec{S} \,. \label{eq:TISEApp}
\end{equation}

Before solving for $ \hat{\psi} $, let us notice that $H$ enjoys the following symmetry for any odd profile $ f(x_2) $ of the angular velocity:
\begin{equation}
	\Pi H \Pi^{-1} = H \,, \label{eq:Symm}
\end{equation}
where 
\begin{equation}
	\Pi \coloneqq \op{diag} (1,1,-1) \otimes \Pi_2 \equiv M \Pi_2 \,, \qquad \Pi^2 = \id \,, \label{eq:ExprPi}
\end{equation}
and $ \Pi_2 $ denotes the parity operator in direction $ x_2 $ (namely $ \Pi_2 g(x_2) = g(-x_2) $ for any function $g$).

By Eq.\,\eqref{eq:Symm}, $H$ preserves the even/odd parity sectors
\begin{equation}
	\Pi \psi = \pm \psi \,,
\end{equation}
a fact that will be employed later. 

Write $ \hat{\psi} $ in components as $ \hat{\psi} = (\eta, u, v) $, and let moreover
\begin{equation}
	\hat{\psi}_\pm \coloneqq \hat{\psi} \rvert_{x_2 \gtrless 0} = (\eta_\pm, u_\pm, v_\pm) \,. \label{eq:PsiSplit}
\end{equation}
Eq.\,\eqref{eq:TISEApp} is explicitly rewritten as
\begin{equation}
	\begin{array}{lcl}
		(H^\pm - \omega) \hat{\psi}_\pm = 0 & \Leftrightarrow & \hat{\psi}_\pm \in \op{Ker} \, (H^\pm - \omega) \\
		& \Leftrightarrow & 
		\begin{cases}
			- \omega \eta_\pm + k_1 u_\pm - \im v_\pm' = 0 \\
			k_1 \eta_\pm - \omega u_\pm \mp \im f v_\pm = 0 \\
			- \im \eta_\pm' \pm \im f u_\pm - \omega v_\pm = 0 
		\end{cases} \,,
	\end{array} \label{eq:pmSystem}
\end{equation}
where $ (\cdot)' = \partial_2 (\cdot) $, $ \xi = \xi (x_2;k_1) \,, \ (\xi = \eta, u,v) $ and the equations must hold for all $k_1 \in \bbR$. Continuity of $ \eta,v $ across $ x_2 = 0 $ follows from the general theory of differential equations, see e.g.\, Chapter 15 in Ref.\,\cite{DiffEqt}. By contrast, no continuity condition is imposed on $u$. 

Translation invariance in the two half-planes $ x_2 \gtrless 0 $ and the requirement of square-summability justify the following ansatz:
\begin{equation}
	\eta_\pm (x_2) = N_\pm \eul^{- \kappa |x_2|} \,, \qquad u_\pm (x_2) = U_\pm \eul^{- \kappa |x_2|} \,, \qquad v_\pm (x_2) = V_\pm \eul^{- \kappa |x_2|} \,, \qquad N_\pm, U_\pm, V_\pm \in \bbC \,, \label{eq:Ansaetze}
\end{equation}
with $ \kappa > 0 $ a positive constant.

Moreover, by continuity of $ \eta,v $
\begin{equation}
	N_+ = N_- \equiv N_0 \,, \qquad V_+ = V_- \equiv V_0 \,.
\end{equation}
Plugging the Ans\"atze \eqref{eq:Ansaetze} into Eq.\,\eqref{eq:pmSystem} leads to
\begin{equation}
	\begin{cases}
		- \omega N_0 + k_1 U_\pm \pm \im \kappa V_0 = 0 \\
		k_1 N_0 - \omega U_\pm \mp \im f V_0 = 0 \\
		\pm \im \kappa N_0 \pm \im f U_\pm - \omega V_0 = 0 \,.
	\end{cases} \label{eq:RedSystem}
\end{equation}
We now consider the even/odd cases $ \Pi \hat{\psi} = \pm \hat{\psi} $ (same as $\Pi \psi = \pm \psi$) separately.
\begin{itemize}
	\item \textbf{Even states.} Observe that
	\begin{equation}
		\Pi \hat{\psi} (x_2) = \hat{\psi} (x_2) \, \longleftrightarrow  \hat{\psi}_+ (x_2) = M \hat{\psi}_- (-x_2) \,, \label{eq:EvenState}
	\end{equation}
	cf.\,(\ref{eq:ExprPi},\ref{eq:PsiSplit}). Component-wise
	\begin{equation}
		\eta_+ (x_2) = \eta_- (-x_2) \,, \qquad u_+ (x_2) = u_- (-x_2) \,, \qquad v_+ (x_2) = - v_- (-x_2) \,,
	\end{equation}
	which combined with Eqs.\,\eqref{eq:Ansaetze} gives
	\begin{equation}
		N_+ = N_- = N_0 \,, \qquad U_+ = U_- \equiv U_0 \,, \qquad V_+ = - V_- \,. 
	\end{equation}
	The first equality is known by continuity of $\eta$. The last one complies with $ v $ continuous if and only if 
	\begin{equation}
		v(x_2) \equiv 0 \ \leftrightarrow \ V_0 = 0 \,. \label{eq:VZeroIsZero}
	\end{equation} 
	System \eqref{eq:RedSystem} reduces to
	\begin{equation}
		\begin{cases}
			- \omega N_0 + k_1 U_0 = 0 \\
			k_1 N_0 - \omega U_0 = 0 \\
			\kappa N_0 + f U_0 = 0
		\end{cases} \ \leftrightarrow \
		\begin{cases}
			\omega N_0 = k_1 U_0 \\
			(k_1^2-\omega^2) U_0 = 0 \\
			(k_1 \kappa + \omega f) U_0 = 0 \,.
		\end{cases}
	\end{equation}
	Besides the trivial one $ N_0 = U_0 = 0 $, there exists exactly one solution meeting our requirement $ \kappa > 0 $, namely
	\begin{equation}
		\kappa = f \,, \qquad \omega = - k_1 \,, \qquad U_0 = - N_0 \,. \label{eq:fKappa}
	\end{equation}
	We have thus recovered the edge channel $ \omega_a (k_1) $ of Eq.\,\eqref{eq:EdgeChannels}. By Eqs.\,(\ref{eq:Ansaetze},\ref{eq:VZeroIsZero},\ref{eq:fKappa}), the corresponding state reads
	\begin{equation}
		\hat{\psi}_a (x_2;k_1) = C_a \eul^{- f |x_2|} 
		\begin{pmatrix}
			1 \\
			- 1 \\
			0
		\end{pmatrix} \,,
	\end{equation}
	with $ C_a \in \bbC$ some normalization constant. $ \psi_a (\ul{x};k_1) $ as in Eq.\,\eqref{eq:EdgeStatesSW} is finally obtained by
	\begin{equation}
		\psi_a (\ul{x};k_1) = \eul^{\im k_1 x_1} \hat{\psi}_a (x_2;k_1) \,.
	\end{equation}

	\item \textbf{Odd states.} The procedure is identical. Enforcing
	\begin{equation}
		\Pi \hat{\psi} (x_2) = - \hat{\psi} (x_2) \, \longleftrightarrow  \hat{\psi}_+ (x_2) = - M \hat{\psi}_- (-x_2) \,, \label{eq:OddState}
	\end{equation}
	produces
	\begin{equation}
		\eta (x_2) \equiv 0 \,, \qquad u_+ (x_2) = - u_- (-x_2) \,, \qquad v_+ (x_2) = v_- (-x_2) \,,
	\end{equation}
	whence 
	\begin{equation}
		N_0 = 0 \,, \qquad U_+ = - U_- \qquad V_+ = V_- = V_0 
	\end{equation}
	and
	\begin{equation}
		\begin{cases}
			k_1 U_+ = - \im \kappa V_0 \\
			(f k_1 - \omega \kappa) V_0 = 0 \\
			(f \kappa - \omega k_1) V_0 = 0 \,.
		\end{cases}
	\end{equation}
	The only non-trivial solution is
	\begin{equation}
		\kappa = | k_1 | \,, \qquad \omega = f \op{sgn} (k_1) \,, \qquad U_+ = - \im \op{sgn} (k_1) V_0 \,.
	\end{equation}
	The middle quantity is indeed the \textit{odd} edge channel $ \omega_b (k_1) $ of Eq.\,\eqref{eq:EdgeChannels}. The corresponding state reads (cf.\,Eq.\,\eqref{eq:OddState})
	\begin{equation}
		\hat{\psi}_b (x_2;k_1) = C_b \eul^{- |k_1| |x_2|} 
		\begin{pmatrix}
			0 \\
			\op{sgn} (x_2) \op{sgn} (k_1) \\
			\im
		\end{pmatrix} \,,
	\end{equation}
	or equivalently
	\begin{equation}
		\psi_b (\ul{x};k_1) = \eul^{\im k_1 x_1} \hat{\psi}_b (x_2;k_1) \,,
	\end{equation}
	as in Eq.\,\eqref{eq:EdgeStatesSW}.
\end{itemize}
\section{Failure of BEC: Dirac model on the half-plane} \label{app:BECFailureDirac}

In this Appendix, we clarify how bulk-edge correspondence is violated in the massive 2D Dirac model on a manifold with boundary. More specifically, we restrict $ H_+ $ (cf.\,Eq.\,\eqref{eq:HPlusMinus}) to the upper half-plane, identify its self-adjoint extensions, compute edge states and their dispersion relation. A single right-moving edge mode appears for half of the self-adjoint boundary conditions. The edge index is accordingly $0$ or $1$, never matching the bulk \enquote{invariant} $ C \hbar (\cal{E}_+) = +1/2 $ found in Eq.\,\eqref{eq:FakeChernPlusMinus}. 

In the interest of brevity, most of the calculations are omitted. A more detailed analysis can be found in \cite{Gruber06}.
\vspace{1\baselineskip} 

Let $ H^\#_+ $ denote the restriction of $ H_+ $ (as per Eq.\,\eqref{eq:HPlusMinus}) to the Hilbert space $ \cal{H}^\# = L^2 (\bbR \times \bbR_+) \otimes \bbC^2 $, where
\begin{equation}
	\bbR \times \bbR_+ = \{ (x_1,x_2) \in \bbR^2 \ | \ x_2 \geq 0 \} \,.
\end{equation}
Translation invariance in direction $ x_1 $ is retained, and the operator $ H^\#_+ $ can be written fiber-wise upon partial Fourier transform $ p_1 \mapsto k_1 $,
\begin{align}
	H^\#_+ (k_1) \coloneqq  k_1 \sigma_1 + (-i \partial_2) \sigma_2 + m_+ \sigma_3 \,, \qquad m_+ > 0 \,. 
\end{align}
This object is symmetric $ \forall k_1 $. Its deficiency indices $ n_\pm $ are equal, $ n_+ = n_- = +1 $, and by von Neumann's theorem (see Chapter 3 of \cite{SAExt}) self-adjoint extensions $H^\#_+ (k_1,z) $ of this operator are parameterized by a unitary $ z \in U(1) \cong S^1$. States
\begin{equation}
	\psi (\ul{x};k_1) = \eul^{\im k_1 x_1} \hat{\psi} (x_2;k_1) \label{eq:StatesCutDirac}
\end{equation} 
in the domain of $H^\#_+ (k_1,z) $ must then satisfy the following boundary conditions, expressed in terms of the phase $ z $: 
\begin{align} 
	\hat{\psi} (x_2 = 0; k_1) \in V_{\mathit{boundary}} (z) \coloneqq \mathrm{span}_\mathbb{C}(v_0), \qquad 
	v_0 = 
	\begin{pmatrix}
		\im + z   \\
		1 + \im z
	\end{pmatrix} \,. 
\end{align}
The boundary conditions, and thus in turn $z$, are assumed independent of $ k_1 $ (the same assumption appears in \cite{Gruber06}, see comment right above Section 3). 

Having finished the characterisation of the set of self-adjoint boundary conditions, we turn our attention to edge states and their dispersion relations. 

We call a vector $ \psi (\ul{x};k_1) $ as in Eq.\,\eqref{eq:StatesCutDirac} an \textit{edge state} if it is a bound eigenstate of some self-adjoint extension $H^\#_+ (k_1,z) $. In more precise terms, it must satisfy the following (for all $k_1$):
\begin{enumerate}
	\item $ \hat{\psi}(0;k_1) \in  V_{\mathit{boundary}} (z) \, $;
	
	\item $ \hat{\psi}(x_2;k_1) \in L^2 (\bbR_+) \otimes \bbC^2 \,$;
	
	\item $ H^\#_+ (k_1,z) \hat{\psi}(x_2;k_1) = \omega (k_1,z) \hat{\psi}(x_2;k_1)$ for some dispersion relation $\omega (k_1,z) \,$.
\end{enumerate}
The eigenstate condition is a linear differential equation in $ x_2 $, thus suggesting the Ansatz 
\begin{align}
	\hat{\psi}(x_2;k_1) = 
	\begin{pmatrix}
		A (k_1)  \\
		B (k_1)
	\end{pmatrix} \eul^{-\im \kappa x_2} \,, \label{eq:Ansatz}
\end{align} 
where $ \kappa $ satisfies
\begin{equation}
	\kappa^2 = \omega^2 -k_1^2 - m_+^2 \,. \label{eq:Kappa}
\end{equation}
Using Ansatz \eqref{eq:Ansatz} and imposing $ \hat{\psi}(x_2;k_1) \in V_{\mathit{boundary}} (z) $ leads to the following dispersion relation: 
\begin{align}
	\omega  & = (- k_1 (1+z^2) + \im m_+  (1-z^2) ) \frac{\Bar{z}}{2} \nonumber \\ 
	& = - k_1 \mathrm{Re}(z) + m_+ \, \mathrm{Im}(z) \,, \label{eq:Dispersions}
\end{align}
and accordingly
\begin{equation}
	\kappa = - \im \left( k_1 \op{Im} (z) + m_+ \op{Re} (z) \right) \,. \label{eq:KappaExplicit}
\end{equation}
Square-summability of the wavefunction requires $\kappa \in i \mathbb{R}_-$. Imposing $ \op{Im} \kappa \leq 0  $ onto the $\kappa$ of Eq.\,\eqref{eq:KappaExplicit} yields the condition
\begin{align}
	\mathrm{Im}(z) k_1 \geq - \mathrm{Re}(z) m_+ \,. \label{eq:KappaGeqZero}
\end{align}
	
Let us now identify the \textit{merging points} $ (k_1^\star (z), \omega^\star (z)) $, where the edge channels of Eq.\,\eqref{eq:Dispersions} meet the bulk bands. Such merging points are signaled by $ \kappa = 0 $, i.e.\,by delocalization of the bound states. If $ \kappa = 0$, Eq.\,\eqref{eq:Kappa} states $\omega (k_1,z)^2 = k_1^2 + m_+^2$. Combining this with Eq.\,\eqref{eq:Dispersions} yields
\begin{align}
    k_1^\star (z) = -\frac{\mathrm{Re}(z)}{\mathrm{Im}(z)} m_+ \,, \qquad
    \omega^\star (z) = \frac{1}{\mathrm{Im}(z)} m_+ \,. \label{eq:MergingPoints}
\end{align}
The edge index equals the signed number of merging points meeting the \textit{positive} band, namely with $ \omega^\star (z) > 0 $. By Eq.\,\eqref{eq:MergingPoints}, a single bound state \textit{being born} from the positive energy band exists if $ \op{Im} (z) > 0 $. In such cases, $ \cal{I}^\# = 1 $. By contrast, $ \op{Im} (z) < 0 \Rightarrow \omega^\star (z) < 0 $ and the point spectrum only encounters the negative bulk band, giving no net contribution to the edge index, which thus amounts to $ \cal{I}^\# = 0 $.

The violation of bulk-edge correspondence illustrated in the previous paragraph is also made explicit in the table below.
\begin{table}[h!]
	\centering
	\begin{tabularx}{0.6\linewidth}{Y|Y|Y}
		 & $ \op{Im} (z) > 0 $ & $ \op{Im} (z) < 0 $ \\ \hline \hline
		$\mathrm{C}\hbar(\cal{E}_+)$ & $ + 1/2 $   & $ + 1/2 $ \\ \hline
		$\op{deg} \vec{e}^{\, \infty}_+ $ & $+1$  & $+1$ \\  \hline 
		$ \cal{I}^\#$ & $ +1 $ & $ 0 $ 
	\end{tabularx}
	\caption{Values of the different indices for the positive energy bundle. The degree of $\vec{e}^{\, \infty}_+$ should be read in the light of remark \ref{rem:Degree}.}
	\label{Index_simple}
\end{table}



\end{appendices}

\section*{Declarations}

\subsection*{Funding and/or Conflicts of interests/Competing interests} 

The work presented in this submission is entirely original and has not been previously published, nor is it currently under consideration for publication elsewhere. Moreover, we have taken great care to avoid any form of plagiarism. To the best of our current knowledge, all results and ideas obtained from external sources have been appropriately acknowledged and cited. Finally, we wish to declare that there are no conflicts of interest related to this research, and no specific funding was received to support this study. 

\subsection*{Data availability statement}

Data sharing not applicable to this article as no datasets were generated or analysed during the current study.

\bibliography{bib}

\end{document}